\documentclass[aps, pra, twocolumn, floatfix,amsmath,superscriptaddress]{revtex4}

\usepackage[utf8]{inputenc} 

\usepackage{hyperref}
\usepackage{bm}
\usepackage{color}
\usepackage[dvipsnames]{xcolor}
\usepackage{enumerate}
\usepackage{marvosym}
\usepackage{mciteplus}

\usepackage{graphicx} 
\usepackage{qcircuit}

\usepackage{multirow} 
\usepackage{dcolumn}

\usepackage{amssymb}
\usepackage{amsfonts}
\usepackage{amsmath}
\usepackage{amsthm}
\usepackage{mathtools}

\usepackage{multirow} 

\definecolor{carmine}{rgb}{0.59, 0.0, 0.09}
\usepackage{ulem} 
\usepackage{cancel}

\usepackage[caption=false]{subfig}

\DeclareMathOperator*{\argmax}{arg\,max}

\newcommand{\mathbbm}[1]{\text{\usefont{U}{bbm}{m}{n}#1}}
\newcommand{\bea}{\begin{eqnarray}}
\newcommand{\eea}{\end{eqnarray}}

\newcommand{\densitymatrices}{\mathcal{D}(\mathcal{H})}

\newcommand{\stkout}[1]{\ifmmode\text{\sout{\ensuremath{#1}}}\else\sout{#1}\fi}

\def\C{\hbox{$\mit I$\kern-.7em$\mit C$}}
\def\R{\hbox{$\mit I$\kern-.6em$\mit R$}}
\def\N{\hbox{$\mit I$\kern-.6em$\mit N$}}
\def\ket#1{|#1\rangle}
\def\tr{\mathrm{tr}}
\def\ket#1{\left| #1\right>}
\def\bra#1{\left< #1\right|}

\newtheorem{theorem}{Theorem}

\newtheorem{lemma}[theorem]{Lemma}

\newtheorem{observation}[theorem]{Observation}

\begin{document}

\title{Approximate and ensemble local entanglement transformations for multipartite states}
\date{\today}
\author{David Gunn}
\affiliation{Institute for Theoretical Physics, University of Innsbruck, A–6020 Innsbruck, Austria}
\author{Martin Hebenstreit}
\affiliation{Institute for Theoretical Physics, University of Innsbruck, A–6020 Innsbruck, Austria}
\author{Cornelia Spee} 
\affiliation{Institute for Theoretical Physics, University of Innsbruck, A–6020 Innsbruck, Austria}
\author{Julio I. de Vicente} 
\affiliation{Departamento de Matemáticas, Universidad Carlos III de Madrid, Avda. de la Universidad 30, E-28911, Leganés (Madrid), Spain}
\affiliation{Instituto de Ciencias Matemáticas (ICMAT), E-28049 Madrid, Spain}
\author{Barbara Kraus}
\affiliation{Institute for Theoretical Physics, University of Innsbruck, A–6020 Innsbruck, Austria}
\affiliation{Department of Physics, QAA, Technical University of Munich,
James-Franck-Str. 1, D-85748 Garching, Germany}

\begin{abstract}
Understanding multipartite entanglement is a key goal in quantum information. Entanglement in pure states can be characterised by considering transformations under Local Operations assisted by Classical Communication (LOCC). However, it has been shown that, for $n\ge5$ parties, multipartite pure states are generically isolated, i.e., they can neither be reached nor transformed under LOCC. Nonetheless, in any real lab, one never deterministically transforms a pure initial state exactly to a pure target state. Instead, one transforms a mixed state near the initial state to an ensemble that is on average close to the target state. This motivates studying approximate LOCC transformations. After reviewing in detail the known results in the bipartite case, we present the gaps that remain open in the multipartite case. While the analysis of the multipartite setting is much more technically involved due to the existence of different SLOCC classes, certain features simplify in the approximate setting. In particular, we show that it is sufficient to consider pure initial states, that it is sufficient to consider LOCC protocols with finitely-many rounds of communication and that approximate transformations can be approximated by ensemble transformations within an SLOCC class. Then, we formally define a hierarchy of different forms of approximate transformations that are relevant from a physical point of view. Whereas this hierarchy collapses in the bipartite case, we show that this is not the case for the multipartite setting, which is fundamentally richer. To wit, we show that optimal multipartite approximate transformations are not generally deterministic, that ensemble transformations within an SLOCC class can achieve a higher fidelity than deterministic transformations within an SLOCC class, and that there are approximate transformations with no deterministic transformations nearby.
\end{abstract}

\maketitle

\section{Introduction}
Entanglement lies at the heart of quantum information theory \cite{Horodeckis2009_QuantumEntanglement}. A key insight into entanglement is that it cannot be created with local operations. As a result, entanglement can be studied through the ``distant labs''  model in which spatially separated parties are each constrained to be able to act only on their local system but are allowed to coordinate their actions by communicating the results of any measurements. The set of operations achievable in this setting is referred to as Local Operations assisted by Classical Communication (LOCC) \cite{DonaldHorodeckiRudolph2002_LOCC}. From this perspective, entanglement becomes a resource which enhances the information processing power of the separated parties  (see e.g. Refs \cite{Horodeckis2009_QuantumEntanglement,PlenioVirmani2007_IntroToEntanglementMeasures}). As entanglement cannot be created locally, if $\ket{\psi}\rightarrow_{\text{LOCC}}\ket{\phi}$ then $E(\ket{\psi})\ge E(\ket{\phi})$ for any measure of entanglement, $E$. Note, it may be the case that  $\ket{\psi}\not\rightarrow_{\text{LOCC}}\ket{\phi}$ nor $\ket{\phi}\not\rightarrow_{\text{LOCC}}\ket{\psi}$. In this case, the entanglement of the states cannot be compared. Thus, LOCC transformations induce only a partial order in the Hilbert space. Nonetheless, in cases where LOCC transformations are possible, this resource theoretic approach to studying entanglement gives us a physically motivated means to quantitatively compare entanglement.

Bipartite entanglement is very well understood. Ref. \cite{Nielsen1999_NielsensThm} provided necessary and sufficient
conditions for LOCC transformations. As a result, one can identify a unique (up to Local Unitaries [LUs]) maximally entangled state, capable of reaching the entire Hilbert space. In the three qubit case, contrary to the bipartite case, not all fully-entangled states (pure states with maximal local rank) can be converted to one another with some non-zero probability \cite{DurEtAl2000_ThreeQubitSLOCC}. Instead, fully-entangled three qubit states are partitioned into two Stochastic LOCC (SLOCC) equivalence classes with incomparable forms of entanglement. Nonetheless, Ref. \cite{Turgut2010_GHZtransfo, KintaTurgut2010_WStateTransfo} successfully characterised state conversions for three qubits. Whilst there is no single maximally entangled state, Ref. \cite{DeVicenteEtAl2013_MES} identified a zero-measure Maximally Entangled Set (MES) - a minimal set of states, capable of reaching the entire Hilbert space. However, for five or more qubits, Ref. \cite{GourKrausWallach2017_AlmostAllTrivStab, SauerweinEtAl2018_AlmostAllStatesNotReachable} showed that states are generically isolated under LOCC. That is, for almost all states, $\ket{\psi}$, there neither exists a fully-entangled state that can be transformed to $\ket{\psi}$ via LOCC, nor can $\ket{\psi}$ be transformed with LOCC into any other fully-entangled state (up to LUs). Consequently, the MES is full-measure, and the partial order induced by LOCC is trivial. Thus, generically, the entanglement of pure states cannot be compared.

This result motivates studying physically motivated modifications to the standard entanglement theory picture. One of the most natural modifications is to study approximate transformations instead of exact transformations. Indeed, in any physically realistic situation, one will not be interested in exact, deterministic transformations of pure states but instead approximate, ensemble transformations of mixed states close to pure states. This approximate setting has already been studied and solved in the bipartite case \cite{VidalJonathanNielsen2000_ApproxLOCC}. There it was shown that if the initial state is pure, then the optimal approximate transformation is surprisingly a deterministic transformation to a nearby state. The multipartite setting is less well studied. Ref. \cite{Acin2000_OptimalDistillGHZ} briefly comments on approximate transformations of three qubits, and approximate transformations under ``resource non-generating operations''  (i.e., a class of operations larger than LOCC) have been studied in the context of generalised resource theories \cite{Regula2022_ApproximateNonEntanglingTransfo}. Otherwise, multipartite approximate LOCC transformations have been largely unexplored. This is precisely the setting we study in this work.

The remainder of the paper is structured as follows. In Section I, we introduce our notation and give a summary of known results on state transformations under LOCC (see Table \ref{tab:knownresultssummary}). This also motivates the study of multipartite, approximate LOCC transformations presented here. In Section II, we then precisely define the approximate transformations that we study in this paper. Here, we also show that it is sufficient to consider pure initial states, finite round LOCC protocols and that general faithful transformations can be arbitrarily well approximated by faithful transformations within an SLOCC class. Then, in Section III, we set out to better understand the approximate transformations we have defined. We first show faithful transformations within an SLOCC class are more powerful than deterministic transformations within an SLOCC class. We then consider general faithful transformations and show that, unlike in the bipartite case, optimal multipartite faithful transformations are not generally deterministic. Finally, we investigate the question of whether, for any given faithful transformation, there is always a nearby deterministic transformation. We provide strong numerical evidence that faithful transformations are indeed more powerful than deterministic transformations between the vicinities of the states.

\section{Preliminaries}
\subsection{Preliminary Results on State Transformations}

In this section, we will introduce our notation and will give an overview of known results on state transformations under LOCC. This will allow us to put the results derived in this paper into a wider context. We do this by first reviewing the bipartite results regarding LOCC transformations and how these results were expanded to various approximate transformations. To this end, we review the results on deterministic, probabilistic, ensemble, faithful, finite-copy, catalytic and asymptotic LOCC transformations in the bipartite setting. We then move on to multipartite systems and consider the corresponding known results in each of these different settings. The results reviewed here are summarized in Table \ref{tab:knownresultssummary}. 

Throughout this paper, we will generally use $\ket{\psi}$ to denote the initial state of a transformation and $\ket{\phi}$ ($\{(p_i, \ket{\phi_i})\}_{i\in I}$, with $I$ some index set) to denote the final state (final ensemble) of a transformation. We use $\ket{\psi}\rightarrow_{LOCC}\ket{\phi}$ to denote that $\ket{\psi}$ can be converted to $\ket{\phi}$ via an LOCC protocol. We will refer to a desired output state as the target state. Sometimes we will use mixed states, in which case $\rho$ and $\sigma$ will be typically used for the input and output states respectively. Moreover, given a pure state, $\ket{\psi}$, we denote the corresponding mixed state as $\psi$, i.e., $\psi=|\psi\rangle\langle\psi|$. We will mainly consider the situation where $n$ spatially-separated parties each have access to a local quantum system of dimension $d$. Thus, we study transformations of states in the Hilbert space $\mathcal{H}\cong \mathbbm{C}_d^{\otimes n}$. We denote the set of mixed states on $\mathcal{H}$ by $\mathcal{D}(\mathcal{H})$. Finally, we will use the phrase ``fully-entangled''  to refer to pure states which have maximum local rank on all sites, i.e., $\text{rk}(\rho_i)=d$ for all single-site reduced density matrices, $\rho_i$.

One of the main motivations for studying exact transformations is that the existence of an LOCC protocol transforming one state into the other (deterministically) ensures that the entanglement of the initial state is at least as high as the entanglement of the final state. This holds for any entanglement measure. Hence, the study of LOCC transformations allows one to partially order the entanglement contained in states \footnote{The order is only partial as even within an SLOCC class, two states need not be convertible in either direction via LOCC.}. Local Unitary (LU) operations are the simplest example of LOCC transformations and do not alter the entanglement of states. Indeed, two pure states are reversibly, exactly inter-convertible via LOCC iff they are related by an LU transformation \cite{Gingrich2002_LOCC=LU}. States related by LUs naturally form an equivalence class. The largest equivalence classes under local operations are Stochastic LOCC (SLOCC) classes \cite{Chitambar2014_EverythingYouWantedToKnow}. Two states are said to be SLOCC equivalent if they both can be converted into one another via LOCC with non-vanishing probability in both directions. Thus, the entanglement of SLOCC-inequivalent, fully-entangled states is incomparable in the LOCC framework. For this reason, one usually considers initial and final states that correspond to the same SLOCC class.  Mathematically, two states are called LU equivalent (SLOCC equivalent) if they are related by local unitary (local invertible \cite{DurEtAl2000_ThreeQubitSLOCC}) operators. We write this as $\ket{\psi}\cong_{LU}\ket{\phi}$ ($\ket{\psi}\cong_{SLOCC}\ket{\phi}$). We typically choose a representative for the SLOCC equivalence class, called the ``seed''  state, $\ket{\psi_s}$, and then write (unnormalised) states in the SLOCC class as $g \ket{\psi_s}$, with $g \in GL(d,\mathbbm{C})^{\otimes n}$. 

\subsubsection{Bipartite LOCC State Transformations}
We begin with the bipartite setting. Any pure bipartite state can be written up to LUs as $\ket{\psi}=\sum_{i=0}^{d-1} \sqrt{\psi_i} \ket{i,i}$, where $\psi_{0} \geq \ldots \psi_{d-1}\ge 0$ and $\sum_{i=0}^{d-1} \psi_i =1$ (via its Schmidt decomposition). Thus, $\ket{\psi}$ can be identified uniquely (up to LUs) with its Schmidt vector, $\vec{\psi}=(\psi_{0},...,\psi_{d-1})$. An immediate consequence is that two states are LU equivalent (SLOCC equivalent) iff their Schmidt vectors (the number of non-vanishing Schmidt coefficients, aka their Schmidt rank, $\text{Sr}$) coincide. A deterministic LOCC transformation from $\ket{\psi}$ to $\ket{\phi}$ is possible iff the Schmidt vector of  $\ket{\psi}$ is majorized by the Schmidt vector of $\ket{\phi}$, i.e., $\vec{\psi} \preceq \vec{\phi}$ \cite{Nielsen1999_NielsensThm}. This condition can equivalently be characterised with the entanglement monotones introduced in Ref. \cite{Vidal1999_ProbablisticBipariteMonotones}, $E_l(\ket{\psi})=\sum_{i=l}^{d-1} \psi_i$. Namely, a deterministic LOCC transformation from $\ket{\psi}$ to $\ket{\phi}$ is possible iff $E_l(\ket{\psi})\ge E_l(\ket{\phi}), \forall l\in\{0,...,d-1\}$. These necessary and sufficient conditions allow one to identify a unique (up to LUs) maximally entangled state, $\ket{\Phi^+}=\frac{1}{\sqrt{d}} \sum_{i=0}^{d-1} \ket{i,i}$, capable of  reaching the entire Hilbert space via LOCC. That is, $\ket{\Phi^+}$ can be transformed to any other state in ${\cal H}$ via deterministic LOCC transformations.  

In the event one cannot transform $\ket{\psi}$ to $\ket{\phi}$ deterministically, the transformation might be possible  probabilistically. Such transformations are referred to as ``conclusive''  transformations \cite{VidalJonathanNielsen2000_ApproxLOCC}. In the bipartite setting, such a transformation is possible if and only if the Schmidt rank of $\ket{\phi}$ is less than or equal to that of $\ket{\psi}$, $\text{Sr}(\ket{\psi})\ge \text{Sr}(\ket{\phi})$. In this setting, it is natural to ask what is the maximum success probability of transforming $\ket{\psi}$ to $\ket{\phi}$ with an LOCC protocol, $p_{max}(\ket{\psi}\rightarrow\ket{\phi})$ (which we shorten to $p_{max}$ when the argument is obvious from the context). The entanglement monotones from above can be used to determine $p_{max}$. Namely, $p_{max}$ is given by the minimum of all the ratios of $E_l(\ket{\psi})$ and $E_l(\ket{\phi})$ \cite{Vidal1999_ProbablisticBipariteMonotones}. Moreover, this maximum success probability is achievable with a One Successful Branch Protocol (OSBP), in which only one sequence of measurement outcomes yields the desired state and all others yield states that are no longer fully-entangled.

Conclusive transformations are a subset of the more general set of ensemble transformations. These are transformations in which an initially pure state is transformed into an ensemble of states. In the case of a finite ensemble of pure states, Ref. \cite{JonathanPlenio1999_ReachingEnsemblesofPureStates} showed that $\ket{\psi}$ can be transformed to $\{(p_i, \ket{\phi_i})\}_{i=1}^m$  iff $E_l(\ket{\psi}) \geq \sum_{i=1}^{m} p_i E_l(\ket{\phi_i})$, $\forall l\in\{0,...,d-1\}$.

Another particularly physically relevant subset of ensemble transformations are ``faithful''  transformations \cite{VidalJonathanNielsen2000_ApproxLOCC}. Here one constrains ensemble transformations to those in which an input state, $\ket{\psi}$, is transformed to an output ensemble that has a large enough average fidelity with some desired, pure, target state, $\ket{\phi}$. Recall, the fidelity between two mixed states, $\rho$ and $\sigma$, is given by \cite{Jozsa1994_Fidelity}
\begin{equation}
    \label{eq:fideqn}
   F(\rho,\sigma)=\left( \tr{\sqrt{\sqrt{\rho}\sigma\sqrt{\rho}}}\right)^2
\end{equation}
and is a measure of how indistinguishable two states are \cite{Fuchs1996_Distinguishability} (see also Section \ref{sec:fidelity}). Thus, faithful transformations approximate transformations from $\ket{\psi}$ to $\ket{\phi}$. In Ref. \cite{VidalJonathanNielsen2000_ApproxLOCC}, it was shown that, in the bipartite setting, the optimal fidelity via a faithful transformation (i.e., the maximum achievable fidelity of an output ensemble with respect to the target state) is always achievable with a deterministic transformation.

Note, conclusive and faithful transformations are intimately related. If one can conclusively transform $\ket{\psi}$ to $\ket{\phi}$ with a high success probability, then one can also transform $\ket{\psi}$ to $\ket{\phi}$ faithfully. To see this, it is sufficient to note that for any conclusive transformation from $\ket{\psi}$ to $\ket{\phi}$ with success probability $1-\epsilon$ (where $\epsilon>0$ is small), one can also transform $\ket{\psi}$ via LOCC to the ensemble $\left\{ \big(1-\epsilon,\ket{\phi}\big),\big(\epsilon,\ket{0}^{\otimes n}\big) \right\}$. The average fidelity that this ensemble has with the target state $\ket{\phi}$ is at least $1-\epsilon$. However, the converse does not hold: a faithful transformation does not imply a conclusive transformation, even in the bipartite setting. To see this, consider two states, $\ket{\psi}$ and $\ket{\phi}$ such that $F(\psi,\phi)>1-\epsilon$, yet $\ket{\phi}$ has a higher Schmidt rank than $\ket{\psi}$. $\ket{\psi}$ can be transformed into $\ket{\phi}$ via a faithful transformation trivially (by doing nothing). However, the probability of transforming $\ket{\psi}$ to $\ket{\phi}$ is zero. With this said, if the initial and final state both have the same Schmidt rank, then the optimal faithful transformation (which as mentioned, can be chosen deterministic) outputs a state, $\ket{\chi}$, that also has the same Schmidt rank. Moreover, $p_{max}(\ket{\chi}\rightarrow \ket{\phi})=p_{max}(\ket{\psi}\rightarrow \ket{\phi})$. That is, the maximum success probability is not reduced by first applying the optimal faithful transformation \cite{VidalJonathanNielsen2000_ApproxLOCC}.

Due to physical constraints, it is also relevant to consider  transformations which begin with a state which is not exactly, but only close to, the desired initial state. Ref. \cite{VidalJonathanNielsen2000_ApproxLOCC} extended their results on bipartite faithful transformations to include nearby (potentially mixed) initial states via the inequality
\begin{align}
\label{eq:vidalinequality}
| D(\rho_1 \rightarrow \sigma_1) - D(\rho_2 \rightarrow \sigma_2) | \leq D(\rho_1, \rho_2) + D(\sigma_1, \sigma_2),
\end{align}
where $D$ is the trace distance (which is intimately related to the fidelity -- see Section \ref{sec:fidelity}), and $D(\rho\rightarrow\sigma)$ is the minimum trace distance between a (generally-mixed) state, $\sigma$, and any state to which $\rho$ can be transformed via LOCC. Thus, Eq. (\ref{eq:vidalinequality}) tells us that if we have a known faithful transformation, $\ket{\psi}\rightarrow \{(p_i,\ket{\phi_i})\}_{i\in I}$, with respect to a target state, $\ket{\phi}$, and then we take a (generally mixed) state, $\rho$, in the vicinity of $\ket{\psi}$, i.e., $F(\rho,\psi)\ge1-\delta$ for some $\delta>0$, then the optimal faithful transformation from $\rho$ to $\phi$ will be at most $\sqrt{\delta}$ worse than the original known faithful transformation \footnote{The square root appears due to the relationship between the trace distance and the fidelity, which is discussed further in Section \ref{sec:fidelity} (see Eq. \ref{eq:FidInequality}).}. More precisely we have 
\begin{equation}
    D(\rho\rightarrow \phi) \le D(\psi\rightarrow \phi) + \sqrt{\delta}.
\end{equation}

In all the aforementioned scenarios the initial state was a single copy of a state. However, one can consider more general transformations, where the initial state corresponds to multiple copies of a state, or multiple different states. Exact transformations of finitely many copies of states have been studied \cite{Hardy1999_OptimalConcentration, Bandyopadhyay2002_BipartiteMulticopy,DuanEtAl2005_MulticopyAndCatalysis}, yielding interesting features; for example, multiple copies of a state may be transformable whilst a single copy is not \cite{Bandyopadhyay2002_BipartiteMulticopy}. Expanding from multi-copy transformations to transformations of multiple different states (``multi-state''  transformations) yields further interesting results. The most well-known example of this setting is entanglement catalysis \cite{JonathanPlenio1999_Catalysis, DuanEtAl2005_MulticopyAndCatalysis, Turgut2007_BipartiteCatalysis1, Klimesh2007_BipartiteCataylsis2}, in which a state, $\ket{\psi}$, is transformed into another state, $\ket{\phi}$, with the help of a catalyst, $\ket{\chi}$, i.e., $\ket{\psi} \otimes \ket{\chi} \rightarrow_{LOCC} \ket{\phi} \otimes \ket{\chi}$. Indeed, multi-copy transformations and catalysis were shown to be deeply connected in Ref. \cite{DuanEtAl2005_MulticopyAndCatalysis}. General multi-state transformations have also been studied \cite{FengEtAl2002_MutualCatalysis, NevenEtAl2021_Multicopy}, yielding a yet wider array of possibilities. For instance, it has been shown that a pair of bipartite states $\ket{\psi_1}\otimes\ket{\psi_2}$ can be LU equivalent to another pair of states $\ket{\phi_1}\otimes\ket{\phi_2}$, even though none of the input and output states are LU equivalent, i.e., $\ket{\psi_i} \not\cong_{LU} \ket{\phi_j}$ for $i,j\in \{1,2\}$ \cite{NevenEtAl2021_Multicopy}. Moving away from exact transformations, multi-copy, catalytic and multi-state bipartite transformations have also been studied in the probabilistic setting \cite{Bandyopadhyay2002_BipartiteMulticopy, ChenWinterEtAl2010_TensorRankSLOCCCatalysis} and approximate setting \cite{vanDamHayden2003_Embezzlement}.

Finally, one can consider the asymptotic limit. Here it is well known that $n$ copies of a pure state can be asymptotically, reversibly transformed to $m$ copies of the maximally entangled state at a rate, $m/n$, given by the entanglement of formation of the initial pure state, the von Neumann entropy of the reduced state \cite{BennettEtAl1996_Assymp2ConcentratingPartialEntByLocalOps}. Moreover, the transformation of asymptotically many copies of a mixed state to a state arbitrarily close to a pure state has been widely studied in the literature, where it is known as entanglement distillation (see e.g. \cite{BennettEtAl1996_Assymp1PurificationofNoisyEntandFaithfulTeleport}).

\subsubsection{Multipartite LOCC State Transformations}
\label{sec:multipartiteloccreview}

We now consider the multipartite setting. All the results in the previous section change significantly in the multipartite scenario. This is for two predominant reasons. The first is the increased complexity in the entanglement structure of the states. The second is the increased complexity of LOCC.

The increased complexity in the entanglement structure of states is due to the fact that fully-entangled multipartite states are not necessarily SLOCC equivalent \cite{DurEtAl2000_ThreeQubitSLOCC}. In fact, for four qubits and larger systems, there exist generically infinitely many SLOCC classes \cite{DurEtAl2000_ThreeQubitSLOCC,VerstraeteEtAl2002_FourQubitsSLOCC1}. 
Furthermore, for 3 qudit systems, it has been shown that deciding SLOCC conversion is NP hard \cite{Chitambar2008_SLOCCDecidabilityisNPHard}.

Nonetheless, there has been considerable progress in recent years. To solve the problem of LU equivalence of multipartite states, Ref. \cite{Kraus2010_LUequivalence} introduced an algorithm to determine the local unitaries that relate two $n$-qubit states in the event they are LU equivalent. The problem of SLOCC equivalence has been tackled with the use of SL Invariant Polynomials (SLIPs) (see Ref. 
\cite{EltschkaEtAl2014_SLIPs6} and references therein). These are functions on the Hilbert space that are polynomials with respect to the coefficients of a state and are invariant under determinant-1, local operators, i.e., $f(\otimes g_i \ket{\psi})=f(\ket{\psi}), \forall g_i \in SL(d_i,\mathbbm{C}), \forall \ket{\psi}\in\mathcal{H}$. In particular, SLOCC classes that contain a critical state (a state such that the single party reduced density matrices are maximally mixed for all sites) can be distinguished by ratios of SL-invariant polynomials \cite{GourWallach2013_SLIPs5} \footnote{SLOCC classes which do not contain a critical state have also been studied (see e.g. Ref. \cite{Slowik2020_SLOCCtypes}).}. Moreover, for each such SLOCC class, the critical state is unique (up to LUs) \cite{KempfNess1979} and an algorithm for determining the corresponding critical state from any state in the SLOCC class is provided in Ref. \cite{VerstraeteEtAl_AlgorithmCriticalSLOCC}. The union of all such classes is full-measure ~\cite{KempfNess1979}. That is, the complement is of a lower dimension than the Hilbert space. 

In fact, for $n \ge 5$, the set of states that (a) are SLOCC equivalent to a critical state and (b) do not possess any local symmetry but the trivial one, $\mathbbm{1}^{\otimes n}$, is full-measure \cite{GourKrausWallach2017_AlmostAllTrivStab,SauerweinEtAl2018_AlmostAllStatesNotReachable}. This set of states will play an important role in this paper. We will refer to these states as ``generic''. The lack of symmetry also means that two SLOCC equivalent generic states, $g \ket{\psi_s}$ and $h \ket{\psi_s}$, are LU equivalent iff $G=H$, where $\ket{\psi_s}$ is a representative seed state for the SLOCC class. Here and in the following, $G=g^\dagger g$ and $H=h^\dagger h$ are both local operators \cite{SauerweinEtAl2018_AlmostAllStatesNotReachable}. Not only is this set of generic states full-measure, it is also open and dense (wrt to the standard topology on $\mathcal{H}$).

Whilst some work has been done into characterising state transformations from fully-entangled to non-fully-entangled states (e.g. Refs \cite{GuoChitambarDuan2016_LOCCHightoLowDim1, HebenstreitEtAl_SLOCCHightoLowDim3}), typically one considers transformations within an SLOCC class. The reason for this is that this allows one to order the entanglement contained in these pure states. If we restrict ourselves to considering transformations among fully-entangled states in the same SLOCC class, then we can choose the same seed state, $\ket{\psi_s}$, for both states. We then typically express the initial state as $\ket{\psi}=g \ket{\psi_s} / {n_g}$ and the final state as $\ket{\phi}=h \ket{\psi_s} / {n_h}$, where $g=\otimes_i g_i,\ h=\otimes_i h_i$ with $g_i, h_i \in GL(d,\mathbbm{C})$, and $n_g, n_h$ are normalisation constants. In case one considers more general transformations where one changes SLOCC class, such a description is no longer possible, and the analysis becomes even more involved. This is further complicated by the existence of infinitely many SLOCC classes. For instance, to characterise transformations from a five qubit, fully-entangled state to a four qubit state, i.e., $\ket{\psi}_{12345}\rightarrow_{LOCC}\ket{0}_1\otimes \ket{\phi}_{2345}$, one must deal with the fact that $\ket{\phi}$ can potentially belong to infinitely many different SLOCC classes.  

The second reason characterising LOCC transformations in the multipartite setting is more complicated than the bipartite setting is that the mathematical description of LOCC protocols is considerably more complicated in the multipartite setting \cite{Chitambar2014_EverythingYouWantedToKnow}. Unlike in the bipartite setting, in which one round of communication is sufficient for any LOCC transformation among pure states \cite{LoPopescu2001_BipartiteLOCC=LOCC1}, multipartite LOCC transformations may require infinitely many rounds of communication \cite{Chitambar2011_InfiniteRoundLOCC}, as well as probabilistic intermediate steps \cite{SpeeEtAl2017_LOCCN}. 

Despite these challenges, deterministic transformations of three qubit states have been completely characterised  \cite{Turgut2010_GHZtransfo,KintaTurgut2010_WStateTransfo,DeVicenteEtAl2013_MES}. In three qubit Hilbert spaces, one finds that, instead of a single maximally entangled state, one must settle for a Maximally Entangled Set (MES) \cite{DeVicenteEtAl2013_MES}. This is a minimal set of states with the property that all states in the Hilbert space can be reached deterministically via LOCC from a state inside the MES. For three qubits, the MES is zero-measure. The MES for four qubits has also been studied \cite{SpeeEtAl2016_FourQubitMES}. For $n\ge5$ qubits and $n\ge 4$ qudits, the MES has  been shown to be a full-measure set \cite{GourKrausWallach2017_AlmostAllTrivStab,SauerweinEtAl2018_AlmostAllStatesNotReachable}.

One can obtain considerable insight by considering more mathematically-tractable super-sets of LOCC transformations.  One such super-set of LOCC is the set of separable maps (SEP), which are characterized by local Kraus operators. SEP is a strict superset of LOCC \cite{BennettEtAl1999_SEP>LOCC1, KleinmannEtAl2011_SEP>LOCC2, ChitambarEtAl2012_SEP>LOCC3, HebenstreitEtAl2016_SEP>LOCC4}. Nonetheless, it can be used to provide necessary conditions for the existence of LOCC transformations. The necessary and sufficient conditions for the existence of a SEP map between SLOCC equivalent states have been derived in Ref. \cite{Gour2011_SEP} (see also Ref. \cite{HebenstreitEtAl2021_SEP1isnotSEP}). As will become immediately apparent, the local symmetries of states are central here. Given a state, $\ket{\psi}$, we define the stabilizer of $\ket{\psi}$, $S_{\ket{\psi}}$, to be the set of local invertible operators that leave the state invariant, i.e.,
\begin{equation}
    S_{\ket{\psi}}=\{S=\otimes_i S_i \in GL(d,\mathbbm{C})^{\otimes n}\ :\ S\ket{\psi}=\ket{\psi}\}. 
\end{equation}
 
We also define the set of local operators which annihilate the state, i.e.,
\begin{equation}
    N_{\ket{\psi}}=\{N=\otimes_i N_i \in \text{Mat}(d,\mathbbm{C})^{\otimes n}\ :\ N\ket{\psi}=0\}. 
    \label{eq:AnnihilatingOpertators}
\end{equation}
The necessary and sufficient conditions for the existence of a SEP map between SLOCC equivalent states are then given by the following theorem.

\begin{theorem}[\cite{Gour2011_SEP,HebenstreitEtAl2021_SEP1isnotSEP}]
\label{thm:SEP}
The state $g\ket{\psi_s}$ can be transformed to $h\ket{\psi_s}$ via SEP if and only if there exists a finite set of probabilities $\{p_k\}$, local symmetries $\{S_k\} \subseteq \mathcal{S}_{\ket{\psi_s}}$, and local singular operators $\{N_q\}\subseteq$ $\ \mathcal{N}_{g\ket{\psi_s}}$ such that
\begin{equation}
    \frac{1}{r}\sum_k p_k S_k^\dagger HS_k + g^\dagger\sum_q N_q^\dagger N_q g = G,\label{eq:SEP}
\end{equation}
where $r=||h\ket{\psi_s}||^2/||g\ket{\psi_s}||^2$. 
\end{theorem}
 
Let us emphasize here that the operators in the set $\mathcal{N}_{\ket{\psi}}$ annihilate the state $\ket{\psi}$. Consequently, the corresponding measurement outcomes do not occur. Nevertheless, it has been shown that some SEP transformations among pure states only exist if these operators are taken into account \cite{HebenstreitEtAl2021_SEP1isnotSEP}. This is 
due to the fact that the completeness relation (leading to Eq. (\ref{eq:SEP})) for the measurement operators can be satisfied by including these operators but cannot without them. Separable transformations which do not use operators from the set $\mathcal{N}_{\psi}$ are referred to as $\text{SEP}_1$. An explicit example of such a transformation that is possible with SEP but not with SEP$_1$ is provided in Ref. \cite{HebenstreitEtAl2021_SEP1isnotSEP}. Moreover, it is easily seen that if the stabilizer is trivial, i.e., $S_{\ket{\psi_s}}=\{\mathbbm{1}^{\otimes n}\}$, then a transformation of $g\ket{\psi}$ to $h\ket{\psi}$ is possible with SEP iff it is possible with SEP$_1$ iff it is possible with LOCC.

As well as super-sets, one can consider physically relevant subsets of LOCC. From a practical point of view, a particularly relevant subset of LOCC transformations are LOCC transformations which utilize only finitely many rounds of classical communication ($\text{LOCC}_{\mathbbm{N}}$). In this case, simple necessary and sufficient conditions for fully-entangled states to be reachable via $\text{LOCC}_{\mathbbm{N}}$ and/or convertible via LOCC with one round of communication are known \cite{SpeeEtAl2017_LOCCN}. Note, Ref. \cite{HebenstreitEtAl2021_SEP1isnotSEP} also shows that the set of $\text{LOCC}_{\mathbbm{N}}$ transformations between fully-entangled states is a subset of SEP$_1$.

Considering Theorem \ref{thm:SEP}, we see that symmetries are essential for the existence of local (deterministic) transformations among fully-entangled pure states. Indeed, if a state has only trivial symmetries, i.e., $S_{\ket{\psi}}=\{\mathbbm{1}^{\otimes n}\}$, then the only SLOCC equivalent states it can be transformed to via SEP are LU equivalent states \cite{GourKrausWallach2017_AlmostAllTrivStab,SauerweinEtAl2018_AlmostAllStatesNotReachable}. This has dramatic consequences for entanglement theory. As discussed above, the set of generic states is a full-measure set of states with only trivial symmetries. Consequently, almost all states are isolated under LOCC (and SEP). That is, almost all states can neither be transformed via LOCC to another fully-entangled non-LU-equivalent state, nor can they be reached via LOCC from another non-LU-equivalent state. This implies that the partial-order induced by LOCC in multipartite systems is trivial. Furthermore, it implies that the maximally entangled set (MES) is generically a full-measured set \cite{DeVicenteEtAl2013_MES,GourKrausWallach2017_AlmostAllTrivStab,SauerweinEtAl2018_AlmostAllStatesNotReachable}.

The fact that almost all states are isolated under LOCC means that, in the multipartite setting, we must move away from exact transformations. Due to its physical relevance, a natural choice of more general transformations are approximate transformations. The first of such approximate transformations are  multipartite conclusive transformations. Once again three and four qubit systems correspond to a special case. Refs. \cite{Acin2000_OptimalDistillGHZ, VerstraeteEtAl2002_GHZOptimalDistillation} identified the optimal success probability of distilling a GHZ state from a fully-entangled three qubit state, and likewise in Ref. \cite{Yildiz2010_OptimalDistillationOfWState} for the W state. Moreover, upper and lower bounds have been attained for the maximum probability of general state transformations within the three qubit GHZ class (see e.g. \cite{CuiEtAl2010_UpperBoundsonTransfoBetweenGHZStates}) and W class (see e.g.  \cite{KintaTurgut2010_WStateTransfo}). Bounds on the maximum probability have also been studied in the case of four qubits (see e.g. \cite{ Gour2011_SEP}).

For $n\ge 5$, the study of probabilistic transformations of generic states turns out to be simpler than even the bipartite case. In fact, for generic states, it has been shown that the maximal success probability with which one state can be transformed into the other is given by \cite{GourKrausWallach2017_AlmostAllTrivStab, SauerweinEtAl2018_AlmostAllStatesNotReachable}
\begin{align}
    p_{max}\left(\ket{\psi} \rightarrow \frac{\otimes_i h_i \ket{\psi}}{n_h} \right)= \frac{n_h^2}{\Pi_i \mu_{max}(H_i)},
    \label{eq:pmaxgeneric}
\end{align}
where again $n_h$ is the normalisation constant, $H_i=h_i^{\dagger}h_i$, and $\mu_{max}$ is the maximum eigenvalue. As in the bipartite case, the maximum success probability is achievable with a One Successful Branch Protocol (OSBP). Moreover, Ref. \cite{Sauerwein2018_DifferentiableTransfo} characterizes the optimal intermediate states and optimal SLOCC paths. Respectively, these are states that one can first transform the initial state to without reducing the overall success probability of reaching the final state and continuous paths of optimal intermediate states.

Using the simple observation that Eq. (\ref{eq:pmaxgeneric}) needs to coincide with the minimal ratio of all entanglement monotones \cite{Vidal1999_ProbablisticBipariteMonotones}, a complete set of entanglement monotones within a generic SLOCC class has been derived \cite{Sauerwein2018_DifferentiableTransfo}. Namely, given a generic SLOCC class with seed state $\ket{\psi_s}$, the set of functions 
\begin{align}
    E_{\vec{x}}^{\ket{\psi_s}}(g\ket{\psi}/n_g)= \langle \vec{x}| G|\vec{x} \rangle /n_g^2,
    \label{eq:genericentmonotones}
\end{align}
where $\ket{\vec{x}}$ is any product state, are entanglement monotones. These monotones are easy to calculate as $G$ is a local operator. Moreover, only a finite number suffice to completely characterize the entanglement of a generic state given its SLOCC class. It has also been shown that these monotones are invariant (monotonic) under deterministic (ensemble) SEP within a generic SLOCC class \cite{Sauerwein2018_DifferentiableTransfo} (see Appendix \ref{sec:AppendixSepEnsTransfo} for a further discussion). However, they are not invariant under SLOCC. Hence, they allow one to compare entanglement between states in the same SLOCC class.

Let us now consider ensemble transformations. Before discussing the known results regarding multipartite ensemble transformations, let us highlight some challenges. As mentioned before, deterministic pure state transformations are, with a few exceptions (e.g. Refs \cite{GuoChitambarDuan2016_LOCCHightoLowDim1,HebenstreitEtAl_SLOCCHightoLowDim3}), usually considered within the same SLOCC class. However, many natural ensemble transformations output states that are not in the same SLOCC class; for instance, conclusive transformations that obtain the target state with the maximum success probability necessarily output states in different SLOCC classes in the failing branches (if they did not, the success probability could be improved). Thus, for ensemble transformations, both scenarios - transformations where all the states in the output ensemble are in the same SLOCC class as the initial state and transformations where the outputs are not necessarily in the same SLOCC classes - are physically motivated. Which of the two settings is more appropriate depends on the scenario within which one considers the transformation. We will address below both scenarios. As the transformation is no longer deterministic to a single pure state, local symmetries no longer play such an important role. Furthermore, in the second case, singular measurement operators that do not simply annihilate the state need to be taken into account. This fact leads to the advantage that more transformations are possible and to the disadvantage that their characterizations are much more involved. 

Let us first discuss the situation where all the states in the output ensemble are in the same SLOCC class as the initial and the final state. Note that in this case tools from exact LOCC transformations can be (at least partially) employed. In fact, similarly to deterministic transformations, we have the following theorem.\\

\begin{theorem}[\cite{Gour2011_SEP,HebenstreitEtAl2021_SEP1isnotSEP}]
\label{thm:SEPensemble}
The state $g\ket{\psi_s}$ can be transformed to the (finite) ensemble $\{(p_i, h_i\ket{\psi_s})\}$ (with $h_i$ local and invertible) via SEP if and only if there exists a finite set of probabilities $\{p_{ij}\}$, symmetries $\{S_j\} \subseteq \mathcal{S}_{\ket{\psi_s}}$, and $N_q\in\mathcal{N}_{g\ket{\psi_s}}$ such that $\sum_j p_{ij} = p_{i}$
\begin{equation}
    \sum_{ij} \frac{1}{r_{i}} p_{ij} S_{ij}^\dagger H_i S_{ij} + g^\dagger\sum_q N_q^\dagger N_q g = G,\label{eq:SEPens}
\end{equation}
where $r_i=||h_i\ket{\psi_s}||^2/||g\ket{\psi_s}||^2$.
\end{theorem}

This theorem follows directly from Ref. \cite{Gour2011_SEP} and Theorem 1. However, in order to be complete, we present a proof of Theorem 2 in Appendix \ref{sec:AppendixSepEnsTransfo}. Analogously to the deterministic case, we can consider ensemble transformations under SEP$_1$. In this case, we have
\begin{equation}  
    \sum_{ij} \frac{1}{r_{i}} p_{ij} S_{ij}^\dagger H_i S_{ij}  = G.
    \label{eq:sep1ensemble}
\end{equation}
Note that, as in the deterministic case, $\text{LOCC}_\mathbbm{N}$ ensemble transformations within an SLOCC class are also a subset of SEP$_1$ (see Appendix \ref{sec:AppendixSepEnsTransfo}).

Moving onto general ensemble transformations, the set of possible transformations one obtains is considerably more complex. The example from the beginning of this section, in which one considers transformations between five and four qubit states, gives a good example of this. Namely, the output ensemble may include states belonging to infinitely many different SLOCC classes.  This simple example shows that the set of general ensemble transformations is too rich to be analyzed in a general way. 

It is natural then to restrict the set of ensemble transformations by imposing physically motivated constraints. As in the bipartite setting, one of the most natural constraints is to consider faithful transformations, where the output ensemble has a high average fidelity with some desired target state. Note that considering states up to some finite fidelity blurs the aforementioned SLOCC classification. For instance, it is well known that there are states SLOCC equivalent to the GHZ state that are arbitrarily close to the W state. More generally, the set of states with fidelity greater than $1-\epsilon$ with any given state $\ket{\psi}$ may intersect infinitely many SLOCC classes. Nonetheless, there are some partial results in this setting. Ref. \cite{Acin2000_OptimalDistillGHZ} provides some numerical evidence that the optimal faithful transformation to the three qubit GHZ state is, as in the bipartite setting, deterministic. Moreover, approximate transformations have been considered under ``resource non-generating operations''  \cite{Regula2022_ApproximateNonEntanglingTransfo} - that is, the class of transformations which do not create entangled states from non-entangled states. 
Finally, some progress has also  been made in studying multi-copy, catalytic and multi-state transformations \cite{NevenEtAl2021_Multicopy, ChenWinterEtAl2010_TensorRankSLOCCCatalysis, ChenHayashi2011_MulticopyStochasticMultipartite},  as well as

\onecolumngrid 
\quad 
 \begin{table}[h]
 \begin{center}
\begin{tabular}{|c|l|l|l|}
 \hline
    Transformation & Notation & Bipartite & Multipartite \\ \hline
    \multirow{3}{*}{\begin{tabular}{c} SLOCC \\ equivalence  \end{tabular}} & \multirow{3}{*}{$\ket{\psi} \simeq_{SLOCC} \ket{\phi}$}  & \multirow{3}{*}{iff $\text{Sr}(\ket{\psi})=\text{Sr}(\ket{\phi} )$} & \multirow{3}{*}{\begin{tabular}{l l} General: & NP Hard \cite{Chitambar2008_SLOCCDecidabilityisNPHard} \\  Generic: & iff ratios of SLIPs coincide \cite{GourWallach2013_SLIPs5} \end{tabular}}      \\ 
      &    &    &   \\  
      &    &    &   \\ \hline 
    \multirow{3}{*}{\begin{tabular}{c} LU \\ equivalence  \end{tabular}} & \multirow{3}{*}{$\ket{\psi} \simeq_{LU} \ket{\phi}$}  & \multirow{3}{*}{iff $\vec{\psi}=\vec{\phi}$}  & \multirow{3}{*}{\begin{tabular}{l l} Qubits: & Algorithm to find LUs, if  they exist  \cite{Kraus2010_LUequivalence} \\   Generic: & iff same SLOCC class and $G=H$ \cite{SauerweinEtAl2018_AlmostAllStatesNotReachable} \end{tabular}}     \\
      &    &    &  \\ 
      &    &    &  \\ \hline
    \multirow{3}{*}{\begin{tabular}{c} Deterministic \\ LOCC \end{tabular}} & \multirow{3}{*}{$\ket{\psi} \rightarrow_{LOCC} \ket{\phi}$} & \multirow{3}{*}{iff $\vec{\psi}\preceq \vec{\phi}$ \cite{Nielsen1999_NielsensThm}}  & \multirow{3}{*}{\begin{tabular}{l l} General: & Necessary constraints from SEP \cite{Gour2011_SEP,HebenstreitEtAl2021_SEP1isnotSEP}\\  Generic: & Isolated \cite{GourKrausWallach2017_AlmostAllTrivStab, SauerweinEtAl2018_AlmostAllStatesNotReachable}  \end{tabular}}   \\
      &    &    &  \\ 
      &    &    &  \\ \hline
    \multirow{4}{*}{\begin{tabular}{c} Conclusive \\ LOCC \end{tabular}}  & \multirow{4}{*}{\begin{tabular}{l} $\ket{\psi}\rightarrow_{LOCC}$ \\ $\qquad \{(p, \ket{\phi}),...\}$ \end{tabular}}& \multirow{4}{*}{\begin{tabular}{l}$p_{max}=\min_l \frac{E_l(\ket{\psi})}{E_l(\ket{\phi})}$ and \\ achievable with OSBP \cite{Vidal1999_ProbablisticBipariteMonotones} \end{tabular}} & \multirow{4}{*}{\begin{tabular}{l l} General: & Bounds on $p_{max}$ from SEP \cite{Gour2011_SEP} \\ Generic: & $p_{max}(\ket{\psi}\rightarrow h\ket{\psi})=n_h^2/\lambda_{max}(H))$ and \\ & achievable with OSBP \cite{GourKrausWallach2017_AlmostAllTrivStab,SauerweinEtAl2018_AlmostAllStatesNotReachable} \end{tabular}}  \\
      &    &    &  \\ 
      &    &    &  \\ 
      &    &    &  \\ \hline 
     \multirow{3}{*}{\begin{tabular}{c} Ensemble \\ LOCC \end{tabular}}   & \multirow{3}{*}{\begin{tabular}{l} $\ket{\psi}\rightarrow_{LOCC}$ \\ $\qquad \{(p_i, \ket{\phi_i})\}$ \end{tabular}}&\multirow{3}{*}{\begin{tabular}{l} iff $E_l(\ket{\psi}) $ \\   $\quad \geq \sum_{i=1}^{m} p_i E_l(\ket{\phi_i}) \forall l$ \cite{JonathanPlenio1999_ReachingEnsemblesofPureStates} \end{tabular}} & \multirow{3}{*}{\begin{tabular}{l l} General: &  States in output ensemble may belong  \\  &  to $\infty$ different SLOCC classes \end{tabular}}  \\
      &    &    &  \\
      &    &    &  \\\hline
     \multirow{4}{*}{\begin{tabular}{c} Ensemble \\ LOCC within \\ SLOCC Class \end{tabular}} & \multirow{4}{*}{\begin{tabular}{l}$\ket{\psi}\rightarrow_{LOCC} \quad \ \ $ \\ $\qquad \{(p_i, \ket{\phi_i})\} \ st $  \\ $\qquad \ket{\phi_i}\cong_{SLOCC}\ket{\psi}$ \end{tabular}} & \multirow{4}{*}{  (As above \cite{JonathanPlenio1999_ReachingEnsemblesofPureStates})} &  \multirow{4}{*}{\begin{tabular}{l l} General: &  Necessary constraints from SEP (Thm \ref{thm:SEPensemble}) \\ Generic: & only if $E_{\vec{x}}^{\ket{\psi_s}}(\ket{\psi}) \ge \sum_i p_i E_{\vec{x}}^{\ket{\psi_s}}(\ket{\phi_i})$   \\ & $\forall \ket{\vec{x}}$ \cite{Sauerwein2018_DifferentiableTransfo} \end{tabular}} \\
      &    &    &  \\
      &    &    &  \\
      &    &    &  \\\hline
    \multirow{5}{*}{\begin{tabular}{c} Faithful \\ LOCC \end{tabular}}  & \multirow{5}{*}{\begin{tabular}{l}$\ket{\psi}\rightarrow_{LOCC} \quad \ \ $ \\  $\qquad \{(p_i, \ket{\phi_i})\} \ st $  \\ $\sum_i p_i F(\phi_i,\phi)>1-\epsilon$    \end{tabular}}  &\multirow{5}{*}{\begin{tabular}{l}  Optimal transformation \\ is always a deterministic \\ transformation. Moreover, \\ it leads to $p_{max}$ \cite{JonathanPlenio1999_ReachingEnsemblesofPureStates}   \\    \end{tabular}}  &    \multirow{5}{*}{\begin{tabular}{l l} $n=3:\ \ $ &Numerical Results \cite{Acin2000_OptimalDistillGHZ}\end{tabular}  }  \\
      &    &    &  \\
      &    &    &  \\
      &    &    &  \\ 
      &    &    &  \\ \hline
\end{tabular}   
\end{center}
\caption{Summary of some known results on state transformations under LOCC that are particularly relevant for the analysis of approximate transformations. For more details, see main text. List does not include many copy, catalytic, many-state or asymptotic transformations, nor results regarding $n=3,4$. In the multipartite column, we list two different cases. ``General''  refers to all states. ``Generic''  refers to known results with respect to the full-measure set of states referred to as generic in the preliminaries. We see that there are many open questions regarding faithful multipartite transformations. We will address them in the subsequent sections.}
\label{tab:knownresultssummary}
\end{table}
\twocolumngrid

\noindent asymptotic transformations \cite{LindenEtAl1999MREGS2, BennetEtAl2000_MREGS1} in the multipartite setting.

A summary of the known results on state transformations, as reviewed here, is given in Table \ref{tab:knownresultssummary}. We see that, despite their physical relevance, non-asymptotic, faithful transformations under LOCC remain largely unexplored in the multipartite setting. As in any real lab one never exactly transforms a pure state into another pure state but instead transforms some state nearby the desired initial state into an ensemble of states near the desired final state, the faithful setting is physically relevant.  Consequently, we set out to investigate this setting in this paper. However, before doing so, we make precise what we mean by ``nearby''.

\subsection{LU-Optimized Fidelity}
\label{sec:fidelity}

In this section, we make precise our notion of ``nearby'' by introducing the average LU--optimized fidelity of the ensemble with the target state and discussing some of its key properties. The fidelity (see Eq. (\ref{eq:fideqn})) is a measure of how indistinguishable two quantum states are \cite{Fuchs1996_Distinguishability}, with  $F$ being symmetric, basis independent, $F(\rho,\sigma)=1$ iff $\rho=\sigma$, $F(\rho,\sigma)=0$ iff the states have orthogonal support. However, the fidelity is obviously not invariant under LUs applied to only one of the states. Local unitaries do not affect entanglement and can be performed freely and reversibly. Therefore, physically it makes sense to optimise the fidelity over LUs, i.e., we consider
\begin{align}
    F_{LU}(\rho,\sigma)&=\max_{U\in LU} F(\rho , U \sigma U^\dagger).
\end{align}

 Throughout this paper then, $\epsilon$-close is meant in terms of the LU-optimized fidelity, i.e., $\sigma\in\mathcal{D}(\mathcal{H})$ is $\epsilon$-close to $\ket{\phi}$ if $F_{LU}(\sigma,\phi) \geq 1-\epsilon$. Correspondingly, an $\epsilon$-vicinity around $\ket{\psi}$ is the set of (generally mixed) states $\epsilon$-close to $\ket{\psi}$. Note that throughout this paper the $\epsilon$-vicinity  refers to (potentially mixed) states with support that is a subset of $\mathcal{H}$. That is given a $\psi\in\densitymatrices$, we define the $\epsilon$-vicinity as
 \begin{equation}
     \{ \rho \in \densitymatrices : F_{LU}(\rho,\psi)\ge 1-\epsilon\}
 \end{equation}
 
Note that the Hilbert space considered here is not necessarily the smallest dimensional Hilbert space containing $\ket{\psi}$. In particular, when studying transformations to a target state, the underlying Hilbert space, $\mathcal{H}$, is the smallest dimensional Hilbert space containing both the initial and final state. The $\epsilon$--vicinities are defined accordingly.

 An ensemble $\{(p_i,\sigma_i)\}$ is $\epsilon$-close to $\ket{\phi}$, if it is on average $\epsilon$-close, i.e.:
 \begin{equation}
     F_{av}\left(\{(p_i,\sigma_i)\},\phi\right)=\sum p_i F_{LU}(\sigma_i,\phi) \geq 1-\epsilon
 \end{equation}

Note, as LU's can be performed freely when considering LOCC transformations, wlog we can assume that, for a given input, $\rho$, and target state, $\ket{\phi}$, a map, $\Lambda$, is always LU-optimized, i.e., $ \Lambda(\rho)=\{(p_i, \sigma_i) \}$ and $F_{LU}(\sigma_i,\phi)=F(\sigma_i,\phi),\ \forall i$. For LU-optimized maps with a pure target state, we have $F_{av}(\Lambda(\rho),\phi)=F(\sum p_i \sigma_i,\phi)$. Moreover, it follows from the purity of $\phi$ that, for LU-optimized maps and any decomposition $\rho=\sum_j q_j \psi_j$, we have 
\begin{equation}
    F_{av}\big(\Lambda (\rho), \phi \big) \le \sum_j q_j F_{av}\big(\Lambda(\psi_j),\phi\big).
    \label{eq:Favisconvex}
\end{equation}

 We can succinctly summarise approximate transformations, then, as ensemble transformations which map an input state $\delta$-close to some ideal pure initial state, $\ket{\psi}$, to an output ensemble $\epsilon$-close to some ideal pure target state, $\ket{\phi}$.

Finally, let us mention here a well--known relation between the fidelity and the trace distance, $D(\rho,\sigma)=\frac{1}{2}||\rho-\sigma||_1$ (which is also a measure of indistinguishability), namely, \cite{FuchsDeGraaf1997_FidelityInequalities}
\begin{equation}
   1-\sqrt{F(\rho,\sigma)}\le D(\rho,\sigma) \le \sqrt{1-F(\rho,\sigma)} 
   \label{eq:FidInequality}.
\end{equation}
In the event one of the states is pure, the lower bound tightens to 
\begin{equation}
   1-F(\rho,\phi) \le D(\rho,\phi) 
   \label{eq:FidInequalityonepurestate},
\end{equation}
and when both states are pure, the upper bound in Eq. (\ref{eq:FidInequality}) is exact. However, unlike the fidelity, the trace distance has the advantage of being a metric. Furthermore, the LU-optimised trace distance, $\min_{U\in LU} D(U  \rho U^\dagger, \sigma)$ is also a metric (between LU orbits) as can be easily verified \footnote{This can be easily seen as follows. $
    \min_{U\in LU}  D(U  \rho U^\dagger, \sigma) = \min_{U,V\in LU} D(U\rho U^\dagger, V\rho V^\dagger)  \le  D(U_0\rho U_0^\dagger, V_0^\dagger \rho V_0)   \le  D(U_0\rho U_0^\dagger, \eta) +  D(\eta, V_0^\dagger \rho V_0) =\min_{U\in LU} D(U\rho U^\dagger, \eta) + \min_{V\in LU} D(V\eta V^\dagger, \sigma)$
where $U_0, V_0$ are the unitaries which minimise $D(U\rho U^\dagger, \eta)$ and $D(V\eta V^\dagger, \sigma)$ respectively, and we have used the basis independence and metric properties of $D$.}.

\section{General Properties of Approximate  Transformations }

In this section, we define precisely approximate transformations and illuminate some of their general properties. As discussed in the preliminaries, approximate transformations are a subset of ensemble transformations. Consequently, we begin by introducing several sets of physically-motivated types of ensemble transformation. Next, as motivated above, we constrain these ensemble transformations by imposing that the input states and output ensembles are near - wrt to the LU-optimised (average) fidelity - to some ideal, pure initial and target states respectively, thereby restricting our consideration to types of approximate transformations. Having defined precisely approximate transformations, we recap the bipartite results in light of these definitions. We then start to consider the general properties of approximate transformations, which will make our subsequent analysis of the multipartite setting easier. We show that, when studying these transformations, we are in fact justified in restricting ourselves to transformations where the initial state is pure. Furthermore, we show that, unsurprisingly, it is sufficient to consider LOCC protocols with only finitely many rounds of communication. We also prove that approximate transformations within an SLOCC class can approximate arbitrarily well general approximate transformations. 

\subsection{Types of Ensemble Transformations}

To begin, we consider LOCC transformations, $\Lambda$, which transform a given initial pure state, $\ket{\psi}$, to some finite ensemble of pure states, $\{(p_i,\ket{\phi}_i)\}_{i=1}^m$, i.e., $\Lambda(\psi)=\sum_{i=1}^m p_i \ket{\phi_i}\bra{\phi_i}\otimes \ket{i}\bra{i}$, where $\sum_{i=1}^m p_i =1$. Throughout this paper, we make the physically motivated assumption that any measurement performed by a single party during a given round of an LOCC protocol only has finitely many outcomes. To ease the notation, we write $\Lambda(\psi)= \{(p_i, \ket{\phi}_i)\}_{i=1}^m$ for an ensemble transformation. That is, we consider
\begin{align}
    T_{ens}(\psi)&=\big\{\Lambda \in LOCC:  \Lambda(\psi)=\{(p_i,\ket{\phi}_i)\}_{i=1}^m \big\}
    \label{eq:genensembletransfo}.
\end{align}  
Following the discussion in the preliminaries, we also consider the following physically-relevant subsets of $T_{ens}(\psi)$. Firstly, we have deterministic transformations, i.e., those where there is only one state in the output ensemble \footnote{Here and in the following, we consider all states $\ket{\phi_i}$ in the ensemble to be distinct. 
Note, this is not a restriction. An ensemble that has combined all identical outcomes is reachable via an LOCC protocol iff the original ensemble is reachable.}: 
\begin{align}
    T_{det}(\psi)=\big\{\Lambda \in T_{ens}(\psi) : \exists\ \phi : \Lambda(\psi)=\{(1,\ket{\phi})\}\big\}.
\end{align}
Note, the output of these deterministic transformations need not be in the same SLOCC class as $\ket{\psi}$. For instance, any transformation to a product state is always included in $T_{det}(\psi)$.

Secondly, we have optimal conclusive transformations. Here one considers a specific target state, $\ket{\phi}$. Optimal conclusive transformations are then transformations of $\ket{\psi}$ which reach $\ket{\phi}$ with the maximum success probability, $p_{max}(\ket{\psi}\rightarrow \ket{\phi})=p_{max}$; i.e., if $p_{max}(\ket{\psi}\rightarrow \ket{\phi})>0$ then we define
\begin{align}
    T_{pmax}(\psi,\phi)&=\big\{ \Lambda \in T_{ens}(\psi) : \nonumber \\
    & \quad \qquad   \Lambda(\psi) =\{(p_{max},\ket{\phi}),...\}\big\},
\end{align}
and otherwise $T_{pmax}(\psi,\phi)=\emptyset$. In Ref. \cite{Sauerwein2018_DifferentiableTransfo}, the optimal intermediate states -- i.e., those that can be reached from the initial state without reducing the overall maximum probability of transforming to the final state -- have been characterized for the bipartite and multipartite case. For this reason, we will not focus too much on this set of transformations in this paper.

Thirdly, as motivated in the preliminaries, we can consider transformations within an SLOCC class. That is, we consider
\begin{align}
    &T_{ens-SLOCC}(\psi)\nonumber\\
    &\quad =\{ \Lambda \in T_{ens}(\psi) : \Lambda(\psi) =\{(p_i,\ket{\phi_i})\}, \mbox{ with } \nonumber\\
    &\qquad \qquad \qquad \qquad \qquad \qquad \ket{\phi_i} \cong_{SLOCC} \ket{\psi}\  \forall i \}.
\end{align}

Finally, we have deterministic transformations within an SLOCC class (i.e., deterministic transformations excluding transformations to, for example, product states) \footnote{For completeness, we note that the set of conclusive transformations within an SLOCC class, i.e., $T_{pmax-SLOCC}(\psi,\phi)$, is by definition empty. This is because any outputs that are not the target state, $\ket{\phi}$, yet are SLOCC equivalent to $\ket{\phi}$ can be converted to $\ket{\phi}$ with some non-zero probability, thereby contradicting the assumption that the transformation achieves $\ket{\phi}$ with probability $p_{max}(\ket{\psi}\rightarrow\ket{\phi})$.}. We have
\begin{align}
    &T_{det-SLOCC}(\psi)\ = T_{det}(\psi)\cap T_{ens-SLOCC}(\psi). 
\end{align}

\subsection{Approximate and Ensemble Transformations and Optimal Transformations}
\label{sec:definingTsets}

Having identified five types of ensemble transformations of interest, we now restrict these transformations to define the sets of approximate transformations that we study in this paper. In correspondence to the sets of physically relevant transformations introduced above, we define the sets of approximate transformations with respect to the $\delta$ and $\epsilon$-vicinities around the initial and final state respectively. That is, given $\ket{\psi},\ket{\phi}\in\mathcal{H}$, we define the set 
\begin{align}
    T_{ens}^{\delta,\epsilon}(\psi,\phi)& =\{ \Lambda \in T_{ens}(\tilde{\psi}) : |\tilde{\psi}\rangle \in \mathcal{H},\ \nonumber\\
    & \qquad F_{LU}(\tilde{\psi},\psi)\ge 1-\delta,  F_{av}(\Lambda(\tilde{\psi}),\phi)\ge 1-\epsilon \}.
    \label{eq:TensDeltaEpsilon}
\end{align}

The sets $T_{det}^{\delta,\epsilon}(\psi,\phi)$, $ T_{pmax}^{\delta,\epsilon}(\psi,\phi)$ \footnote{To be more precise, naturally $ T_{pmax}^{\delta,\epsilon}(\psi,\phi)$ consists of conclusive transformations in  $T_{ens}^{\delta,\epsilon}(\psi,\phi)$. That is, $T_{pmax}^{\delta,\epsilon}(\psi,\phi)=\{\Lambda\in T_{pmax}(\tilde{\psi},\tilde{\phi}) : F_{LU}(\tilde{\psi},\psi)\ge 1-\delta, F_{LU}(\tilde{\phi},\phi)\ge 1-\epsilon,\ F_{av}(\Lambda(\tilde{\psi}),\phi)\ge 1-\epsilon\} \subseteq T_{ens}^{\delta,\epsilon}(\psi,\phi) $.}, $T_{ens-SLOCC}^{\delta,\epsilon}(\psi,\phi)$ and $T_{det-SLOCC}^{\delta,\epsilon}(\psi,\phi)$ are all defined similarly. These transformations are depicted in Fig. \ref{fig:GeneralTransfo}. Note, that in the case of $\delta=0$, one starts exactly with $\ket{\psi}$.  

For each of these sets of approximate transformations, it is natural to ask what is the optimal transformation. Therefore, we define
\begin{equation}
     F_X^{\delta}(\psi\rightarrow \phi)= \max_{\tilde{\psi} : F_{LU}(\psi,\tilde{\psi})\ge 1-\delta} \sup_{\Lambda\in T_X(\tilde{\psi})} F_{av}(\Lambda(\tilde{\psi}),\phi),
    \label{eq:Fmax}
 \end{equation}
 where $X$ is a stand-in label for each transformation type.

Note, that the optimal fidelity is defined via the supremum as it is not in general the case that the set of transformations is closed. Also, note that, if $\delta=0$, then $F^0_{ens}$ is simply the optimal fidelity achievable via a faithful transformation of $\ket{\psi}$, as studied in the bipartite case in Ref. \cite{VidalJonathanNielsen2000_ApproxLOCC}. We use $F_{ens}$ as shorthand for $F_{ens}^0$. Moreover, note that, due to the presence of the supremum, the optimal fidelities and the sets of transformations are subtly related: given two types of transformations, $X$ and $Y$,  $F_X^{\delta}(\psi\rightarrow\phi)=F_Y^{\delta}(\psi\rightarrow\phi)$ does not imply that $T_X^{\delta,\epsilon}(\psi, \phi)=T_Y^{\delta,\epsilon}(\psi, \phi)$ for $\epsilon=1-F_X^{\delta}(\psi\rightarrow \phi)$; nor does the existence of an $\epsilon>0$ such that $T_X^{\delta,\epsilon}(\psi, \phi)\ne \emptyset$ whilst $T_Y^{\delta,\epsilon}(\psi, \phi)=\emptyset$ imply that  $F_X^{\delta}(\psi\rightarrow\phi)>F_Y^{\delta}(\psi\rightarrow\phi)$.

\begin{figure}
    \subfloat[\label{subfiga}]{\includegraphics[width=.95\linewidth]{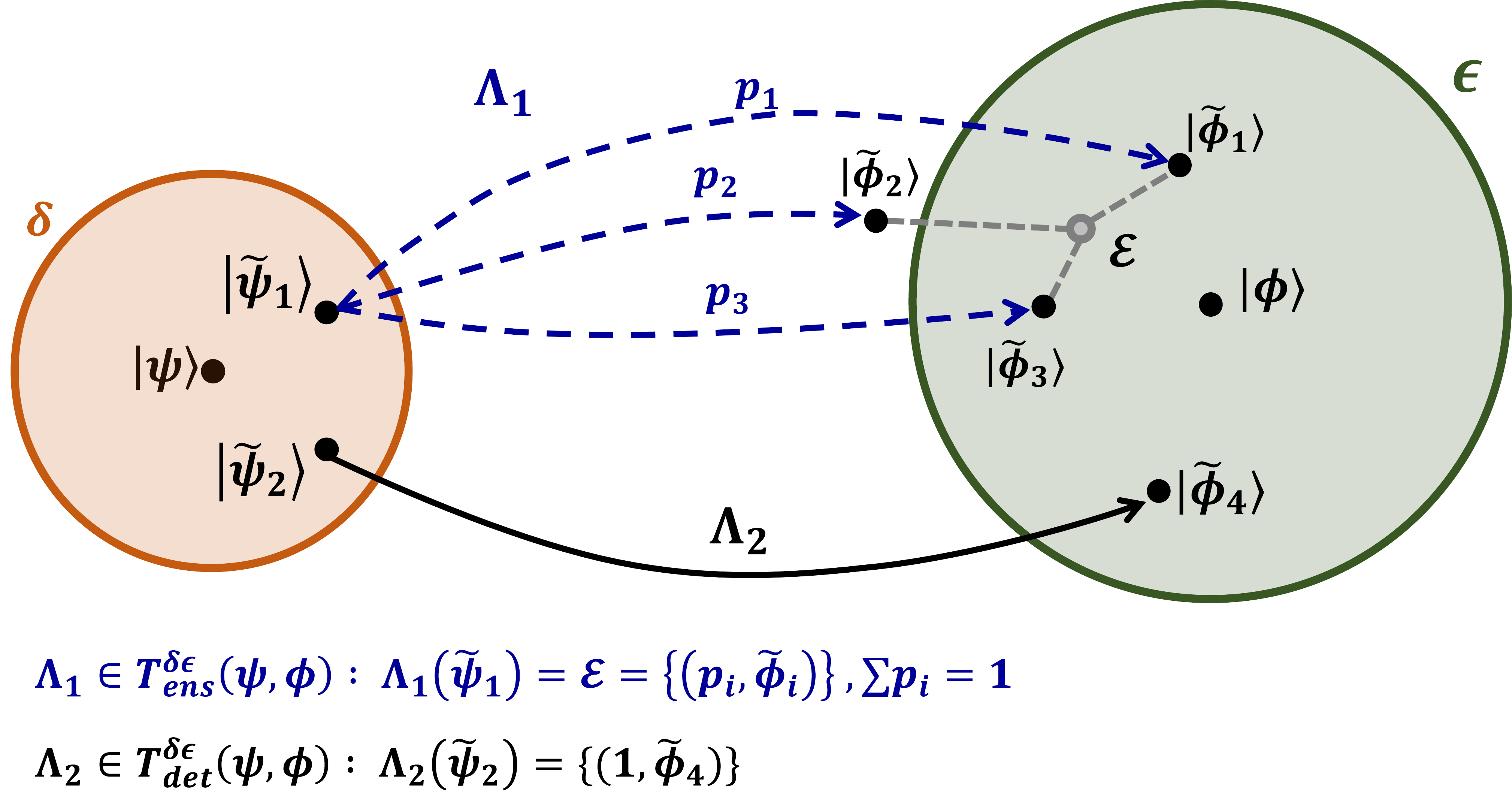}}
    
    \medskip
    
    \subfloat[\label{subfigb}]{\includegraphics[width=.95\linewidth]{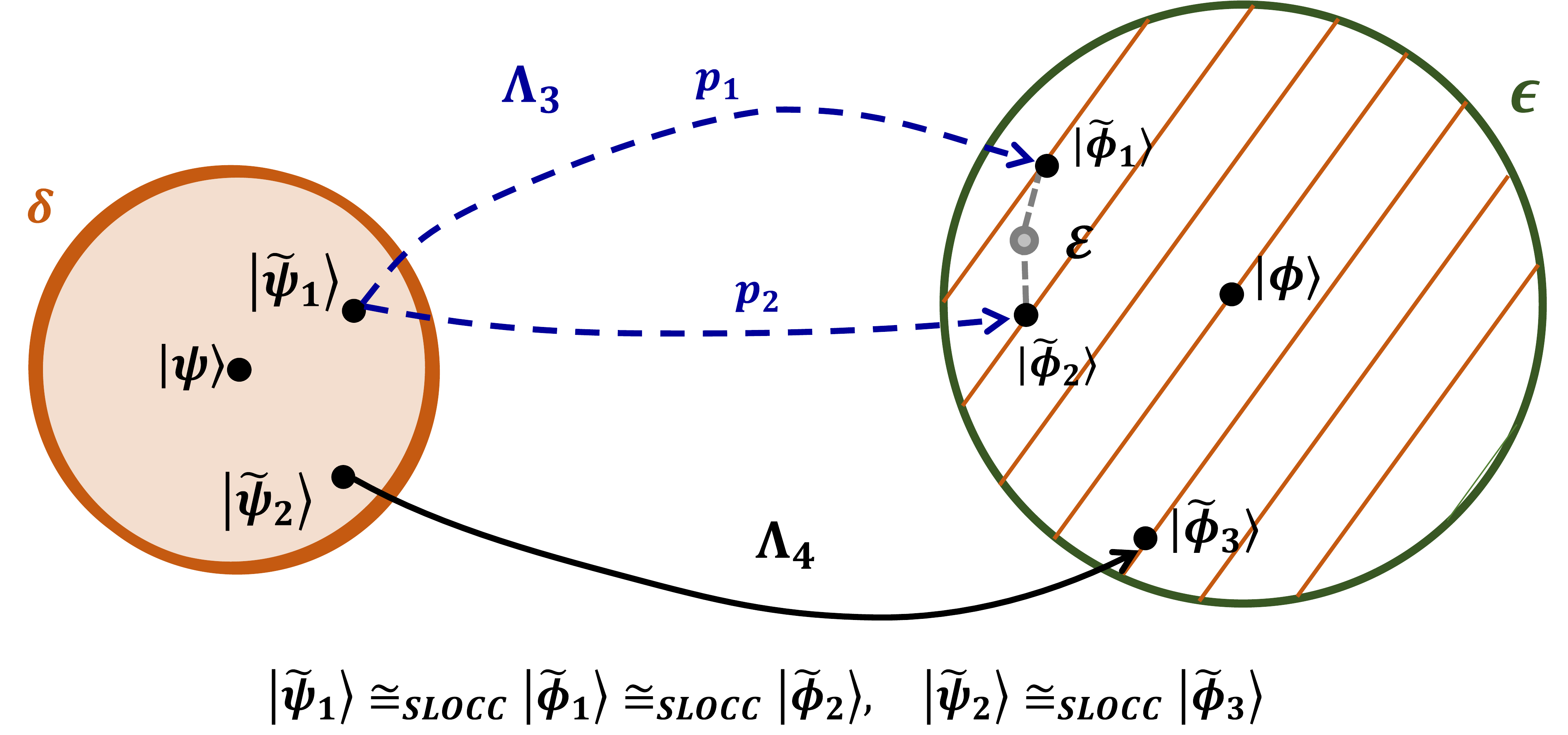}}
    \caption{Approximate transformations: Blue dashed lines represent branches of an ensemble LOCC transformation. Black solid lines represent deterministic LOCC transformations. Fig. \ref{subfiga} concerns general faithful transformations. The orange (green) disk represents the $\delta$ ($\epsilon)$ vicinity around the initial (final) state. $\Lambda_1\in T_{ens}^{\delta,\phi}(\psi,\phi)$ maps a state $|\tilde{\psi}_1\rangle$ (which is $\delta$-near to $\ket{\psi}$) to an ensemble of states, $\{(p_i,|\tilde{\phi}_i\rangle)\}$, such that the average fidelity  with $\ket{\phi}$ (represented by the grey dot) is greater than or equal to $1-\epsilon$. $\Lambda_2\in T_{det}^{\delta,\epsilon}(\psi,\phi)$ is a map that corresponds to a protocol which deterministically transforms $|\tilde{\psi}_2\rangle$ to $|\tilde{\phi}_4\rangle$. Transformation in $T_{pmax}^{\delta,\epsilon}(\psi,\phi)\subseteq T_{ens}^{\delta,\epsilon}$ correspond to transformations with the structure of $\Lambda_1$ but with the additional constraint that $p_i=p_{max}(|\tilde{\psi}\rangle \rightarrow |\tilde{\phi}_i\rangle)$ for one of the outputs in the $\epsilon$-vicinity (i.e., in the above example, either $|\tilde{\phi}_1\rangle$ or $|\tilde{\phi}_3\rangle$). Fig. \ref{subfigb} considers transformations within SLOCC classes. In this case, we only consider output states that are SLOCC equivalent to a state in the initial $\delta$-vicinity (indicated by the incomplete orange fill in the final vicinity). $\Lambda_3 \in T_{ens-SLOCC}^{\delta,\epsilon}(\psi,\phi)$ transforms a state $|\tilde{\psi_1}\rangle$ to an ensemble of states SLOCC equivalent to $|\tilde{\psi_1}\rangle$. $\Lambda_4\in T_{det-SLOCC}^{\delta,\epsilon}(\psi,\phi)$ deterministically transforms $|\tilde{\psi}_2\rangle$ to an SLOCC equivalent state $|\tilde{\phi}_3\rangle$.}
    \label{fig:GeneralTransfo}
\end{figure}

\subsection{Recap of Bipartite Results}

Having defined these sets, let us recap how they relate to one another in the bipartite case. As mentioned, Ref. \cite{VidalJonathanNielsen2000_ApproxLOCC} studied faithful transformations in the bipartite setting. There it was shown that, for any pure state input, the optimal fidelity achievable with a faithful transformation, $F$, can always be achieved with a deterministic transformation to a pure state whose fidelity with the target state is $F$. We can express these results using the terminology introduced above. Firstly, note that if $\delta>0$, then the above results still hold as we simply optimise over pure states $\delta$-close to the initial state. Thus, for all $\delta\ge0$, we have
\begin{equation}
    F_{ens}^\delta(\psi\rightarrow \phi) = F_{det}^\delta(\psi\rightarrow \phi), 
    \label{eq:bipartiteDetIsOptimal}
\end{equation}
where the supremum in Eq. (\ref{eq:Fmax}) is always obtainable. Consequently, whenever an approximate ensemble transformation is possible, an approximate deterministic transformation is also possible; that is,
\begin{equation}
    T_{ens}^{\delta,\epsilon}(\psi,\phi)\ne\emptyset \ \text{iff} \ T_{det}^{\delta,\epsilon}(\psi,\phi)\ne\emptyset.
\end{equation}

We can also consider bipartite SLOCC preserving transformations. Here, we note that the optimal faithful transformation with respect to $\ket{\psi}$ and $\ket{\phi}$, transforms $\ket{\psi}$ to a state with Schmidt rank equal to the minimum of the Schmidt ranks of $\ket{\psi}$ and $\ket{\phi}$ \cite{VidalJonathanNielsen2000_ApproxLOCC}. Now consider $\ket{\psi}\cong_{SLOCC}\ket{\phi}$ to be fully-entangled and consider a state $|\tilde{\psi}\rangle$ $\delta$-close to $\ket{\psi}$ (note, $|\tilde{\psi}\rangle$ need not be fully-entangled but can at most have the same Schmidt rank as $\ket{\psi}$). 
Let $\Lambda\in T_{det}(\tilde{\psi})$ be the deterministic transformation that achieves the optimal fidelity with respect to $\ket{\phi}$. Then $\Lambda$ must output a state that has the same Schmidt rank as $|\tilde{\psi}\rangle$ (as the Schmidt rank of $\ket{\phi}$ is maximal and hence, might be larger than the one of $|\tilde{\psi}\rangle$). Therefore, $\Lambda\in T_{det-SLOCC}(\tilde{\psi})$. This holds for all $|\tilde{\psi}\rangle$ $\delta$-close to $\ket{\psi}$. This is all to say, if $\ket{\psi}\cong_{SLOCC}\ket{\phi}$ are fully-entangled, then
\begin{align}
   F_{ens}^\delta(\psi \rightarrow \phi)&=F_{ens-SLOCC}^\delta(\psi \rightarrow \phi)\\
   &=F_{det-SLOCC}^\delta(\psi \rightarrow \phi).
\end{align}
Moreover, we have, for all $\delta\ge 0$ and $\epsilon \ge 0$,
\begin{align}
  T_{ens}^{\delta,\epsilon}(\psi,\phi)\ne\emptyset \    & 
  \text{iff}\ T_{ens-SLOCC}^{\delta,\epsilon}(\psi,\phi)\ne\emptyset\\
  &\text{iff} \ T_{det-SLOCC}^{\delta,\epsilon}(\psi,\phi)\ne\emptyset.
\end{align}
In this sense, the hierarchy of approximate transformations collapses for fully-entangled states in the bipartite setting.

\subsection{Simplifications of Approximate Transformations}
We now make some general observations which simplify our subsequent analysis in the multipartite setting. Namely, we first show that it is indeed sufficient to study approximate transformations of pure states. Then we show it is sufficient to consider LOCC protocols with finitely-many rounds of communication. Finally, we show that faithful transformations within an SLOCC class approximate general faithful transformations arbitrarily well.

\subsubsection{Mixed vs Pure Initial States}

We begin by showing that it is sufficient to consider transformations with a pure state input. More precisely, we have:

\begin{lemma} Let $\rho$ be $\delta$-close to $\ket{\psi}$ and $\Lambda(\rho)$ be $\epsilon$-close to $\ket{\phi}$, with $\Lambda\in T_{ens}(\psi)$. Then for all $\alpha, \beta >0$ such that $\frac{1}{\alpha}+\frac{1}{\beta} \le 1$ there exists some state $|\tilde{\psi}\rangle$ and transformation $\tilde{\Lambda}\in T_{ens}(\psi)$ such that $|\tilde{\psi}\rangle$ is $(\alpha \delta)$-close to $\ket{\psi}$ and such that $\tilde{\Lambda}(\tilde{\psi})$ is $(\beta \epsilon)$-close to $\ket{\phi}$.\end{lemma}

\proof{
Let $\alpha, \beta >0$ such that $\frac{1}{\alpha}+\frac{1}{\beta} \le  1$, which implies that $\alpha,\beta> 1$. We can assume wlog that $F_{LU}(\rho,\phi)=F(\rho,\phi)$, $\Lambda(\rho)=\{(p_i,\sigma_i)\}$ and $F_{LU}(\sigma_i,\phi)=F(\sigma_i,\phi),\ \forall i$.  Let $\rho$ decompose as $\rho=\sum_{i\in I} q_i\psi_i$. Let us now define a subset $S\subset I$, such that $F(\psi_i,\psi)<1-\alpha\delta$ (i.e., $S$ is the set of states in the decomposition that are outside the $\alpha\delta$-vicinity). Note, naturally this set must be a strict subset (otherwise $F(\rho,\psi)<1-\alpha\delta<  1-\delta$). If $S$ is empty, then it is easy to see that the claim must hold for at least one $\ket{\psi_i}$. If not, then we have
\begin{align}
    1-\delta &\le \sum_{i\in I\setminus S} q_i F(\psi, \psi_i) + \sum_{i\in S} q_i F(\psi,\psi_i)\nonumber\\
    &< q + (1-q)(1-\alpha \delta).
\end{align}
where $q=\sum_{i\in I\setminus S} q_i$. Rearranging, we have
\begin{equation}
    q> 1-\frac{1}{\alpha}.
    \label{eq:ineq1}
\end{equation}
We now complete the proof by showing that at least one of the pure states in $I\setminus S$ has the property that $F_{av}(\Lambda(\psi_i),\phi)\ge 1-\beta \epsilon$. To show this, assume the contrapositive, i.e., $F_{av}(\Lambda(\psi_i),\phi)< 1-\beta \epsilon, \forall i\in I\setminus S$. Then we have
\begin{align}
    1-\epsilon\ &\le F_{av}(\Lambda(\rho),\phi) \\
                &\le \sum_{i\in I\setminus S} q_i F_{av}(\Lambda(\psi_i),\phi)+ \sum_{i\in S} q_i F_{av}(\Lambda(\psi_i),\phi)\\
                & < q (1-\beta\epsilon) + (1-q),
\end{align}
where we have used Eq. (\ref{eq:Favisconvex}) to reach the second line. Rearranging, we have
\begin{equation}
    q < \frac{1}{\beta}.
    \label{eq:ineq2}
\end{equation}
Combining Eq. (\ref{eq:ineq1}) and Eq. (\ref{eq:ineq2}), we deduce $\frac{1}{\alpha}+\frac{1}{\beta}>1$, which contradicts our assumptions. Thus, there must exist a pure state $\ket{\psi_i}\in I\setminus S$ that is $\alpha \delta$-close to $\ket{\psi}$ and is mapped $\beta \epsilon$-close to $\ket{\phi}$. $\qedsymbol$}

This lemma tells us that approximate transformations with mixed state inputs can be studied by considering approximate transformations with pure state inputs but with slightly bigger $\delta$ and $\epsilon$-vincinities ($\alpha=\beta=2$ is sufficient). This is why, when defining $T_{ens}(\psi)$, we consider pure state inputs. Note that a less tight constraint follows directly from Eq. (\ref{eq:vidalinequality}) (see Appendix \ref{sec:AppendixusingVidalsInequality}).

\subsubsection{LOCC vs $\text{LOCC}_\mathbbm{N}$}
\label{sec:LOCCvsLOCCn}

Not only is it sufficient to restrict the initial state to be pure, but it is also sufficient to restrict ourselves to LOCC transformations with only finitely many rounds of communication, $\text{LOCC}_{\mathbbm{N}}$ \cite{Chitambar2014_EverythingYouWantedToKnow, HebenstreitEtAl2021_SEP1isnotSEP}. To see this, let us first recall the definition of an (infinite)-round LOCC protocol \cite{Chitambar2014_EverythingYouWantedToKnow, HebenstreitEtAl2021_SEP1isnotSEP}. A map, $\Lambda$, is in LOCC, if there exists a sequence of $\text{LOCC}_{\mathbbm{N}}$ protocols with corresponding maps $\left(\Lambda_i\right)_{i \in \mathbb{N}}$ with the following three properties: (a) the $i$th protocol consists of $i$ rounds; (b) the first $i$ rounds in the $(i+1)$th protocol coincide with the $i$th protocol; and (c) the sequence $\left( \Lambda_i \right)_{i\in\mathbbm{N}}$ converges to $\Lambda$ with respect to the diamond norm \cite{Chitambar2014_EverythingYouWantedToKnow}. This implies that, for all states $\ket{\psi}$, $\lim_{i \rightarrow \infty}D(\Lambda_i(\ket{\psi}\bra{\psi}), \Lambda(\ket{\psi}\bra{\psi})) = 0$, where $D$ is the trace distance. This means that for all $\eta>0$, there exists an $N \in \mathbb{N}$ such that, for all $i \geq N$, $D(\Lambda_i(\ket{\psi}\bra{\psi}), \Lambda(\ket{\psi}\bra{\psi}))< \eta$. 

Applying this notion to approximate transformations, it follows that for any infinite round protocol that gets $\epsilon$-close to a target state, $\ket{\phi}$, and $\forall \eta>0$, the protocol can be truncated to a finite number of rounds, such that (by the metric property of the trace norm) the truncated protocol yields a state that is $(\epsilon + \eta)$-close to $\ket{\phi}$. Hence, up to arbitrarily small $\eta$, it is sufficient to study $\text{LOCC}_{\mathbbm{N}}$. Moreover, as we assume throughout this paper that each measurement performed by a party during a round of an LOCC protocol only has finitely many outputs, the LOCC$_\mathbbm{N}$ protocol has only finitely many outputs. Thus, we see why we may consider finite ensembles as outputs when defining $T_{ens}$.

\subsubsection{Ensemble Transformations within an SLOCC Class vs General Ensemble Transformations}

We now show that in fact faithful transformations within an SLOCC class, i.e., $\Lambda\in T_{ens-SLOCC}(\psi)$, can approximate general faithful transformations, i.e., $\tilde{\Lambda}\in T_{ens}(\psi)$, arbitrarily well. Specifically, we prove the following theorem.

\begin{theorem}
    \label{thm:Fens=FensSLOCC}
    For all states, $\ket{\psi},\ \ket{\phi}$, and $\delta>0$
    \begin{equation}
        F_{ens}^\delta(\psi,\phi) = F_{ens-SLOCC}^\delta(\psi,\phi)
    \label{eq:Fens=FensSLOCC}
    \end{equation}
\end{theorem}

\begin{proof}
   Before proving this theorem, note that the equality in Eq. (\ref{eq:Fens=FensSLOCC}) clearly relies on the fact that the optimal fidelity is defined with respect to the supremum. Thus, we will show that, for any ensemble transformation which achieves a given average fidelity, there is an SLOCC preserving transformation that achieves an arbitrarily close average fidelity.
    
   Now, consider $n$-qudit  systems and set $\delta=0$ (we extend the argument to $\delta>0$ at the end). By Section \ref{sec:LOCCvsLOCCn}, for all $\epsilon>0$ there exists a finite round LOCC protocol \footnote{Which, as we always assume, has finitely many measurement outcomes per measurement round} with corresponding map, $\Lambda^{\epsilon}$, such that $F_{ens}(\psi,\phi)-F_{av}(\Lambda^\epsilon(\psi),\phi)\leq\epsilon$. Any such finite-round protocol is equivalent to another protocol in which all parties only ever apply two outcome measurements \cite{AnderssonOi2007_BinaryMeasurementsSufficient} (equivalent in the sense that for any input, the new protocol outputs the same outputs with the same probabilities). Therefore, wlog let $\Lambda^\epsilon$ correspond to such a protocol for each $\epsilon$. For any two-outcome measurement during this protocol with measurement operators $\{M_0,M_1\}$, we can use the polar decomposition to write $M_i = U_i Q_i$, with $U_i$ unitary and $Q_i\ge 0$. It follows from the completeness relation that $Q_0$ and $Q_1$ share a common eigenbasis. Thus we may write $M_i=U_i V D_i V^{\dagger}$, where $V$ is unitary and $D_i\ge 0$ is a positive, diagonal matrix. Note, we have two possibilities for each two outcome measurement: (a) both measurements are full-rank, (b) at least one $M_i$ is not full rank (at least one $D_i$ has at least one zero diagonal entry). At the end of the protocol, there will be finitely many outputs. We write $\Lambda^\epsilon=\sum_{i=1}^m \Lambda_i^\epsilon\otimes |i\rangle\langle i|$, where $\Lambda_i^{\epsilon}$ are the CP maps corresponding to the $m\in \mathbbm{N}$ final outputs \footnote{Note, in the protocol, we also implicitly assume that the parties never ``forget''  measurement outcomes. Thus, $\Lambda_i^{\epsilon}$ are Kraus-rank 1 CP maps.}.

We construct now an SLOCC-preserving protocol with outputs arbitrarily close to the original outputs and hence leading to the same average fidelity. To do so, we now modify each non-invertible measurement in the protocol above by an arbitrarily small amount, ensuring that all measurement operators are invertible. To this end, for each $\epsilon>0$, we consider a family of alternative finite-round protocols, $\Lambda^{\epsilon,\chi}$, parameterised by $\chi\in(0,1)$, by replacing each measurement above accordingly: in case (a), we retain the original measurement; in case (b), we implement instead the measurement $\{\tilde{M}_i^\chi\}$ with $\tilde{M}_i^\chi = U_i V \left(\sqrt{(1-\chi) D_i^2 + \chi D_{i\oplus 1}^2} \right) V^\dagger$ (with $\oplus$ indicating addition modulo 2). It is easy to verify these measurement operators satisfy the completeness relation $\forall\chi\in(0,1)$. Moreover, it is clear that $\tilde{M}_i^\chi$ are invertible $\forall\chi\in(0,1)$, as the diagonal entries of $D_0$ and $D_1$ are greater than or equal to zero and never vanish for the same entry (otherwise $\{M_0,M_1\}$ would not be complete). Therefore, $\Lambda^{\epsilon,\chi}\in T_{ens-SLOCC}(\psi)$, $\forall \chi\in (0,1)$.
   
    For the final step of the argument, note that, since there are finitely many rounds, for each $\epsilon$, each of the outputs of $\Lambda^{\epsilon,\chi}$, $\{\Lambda_i^{\epsilon,\chi}(\psi)\}_{i=1}^m$, is continuous in $\chi$ \footnote{Note, $\Lambda_i^{\epsilon,\chi}(\psi)$ are not normalised.}. In particular, $\forall\epsilon'>0$, $\exists\chi>0$ such that $D(\Lambda^{\epsilon,\chi}_i(\psi),\Lambda_i^\epsilon(\psi))\leq\epsilon'$, $\forall i \in \{1,..,m\}$. Moreover, $p_i F(\Lambda_i^{\epsilon,\chi}(\psi)/p_i,\phi) \equiv \tr(\Lambda_i^{\epsilon,\chi}(\psi)\ket{\phi}\bra{\phi})$ is also continuous in $\chi$ and thus so is $F_{av}(\Lambda^{\epsilon,\chi}(\psi),\phi)$.  Therefore, for each given $\epsilon$, it holds that $\forall\epsilon''>0$ $\exists\chi>0$ such that $|F_{av}(\Lambda^{\epsilon,\chi}(\psi))-F_{av}(\Lambda^\epsilon(\psi),\phi)|\leq\epsilon''$. 
    
    Thus, for any $\eta>0$, choose $\epsilon,\epsilon''>0$ such that $\epsilon+\epsilon''\leq\eta$, and then
    \begin{align}
        |F_{ens}(\psi,\phi)&-F_{av}(\Lambda^{\epsilon,\chi}(\psi))|\nonumber\\
            &\leq |F_{ens}(\psi,\phi)-F_{av}(\Lambda^\epsilon(\psi),\phi)| \nonumber\\
            &\qquad + |F_{av}(\Lambda^\epsilon(\psi),\phi)-F_{av}(\Lambda^{\epsilon,\chi}(\psi))| \nonumber\\
            &\leq\epsilon''+\epsilon\leq\eta,
    \end{align}
    which proves the result for $\delta=0$. 
    
    To extend to the $\delta>0$ case, it is sufficient to consider the maximum achievable fidelity from any state within the $\delta$-vicinity. Let $|\tilde{\psi}\rangle$ with $F(\tilde{\psi},\psi)\ge 1-\delta$ be the input state for this optimal transformation. Then the above argument applied to this transformation proves the claim.

\end{proof}

As a final comment, we recall that $\text{LOCC}_{\mathbbm{N}}$ ensemble transformations within an SLOCC class are a subset of SEP$_1$ ensemble transformations \cite{HebenstreitEtAl2021_SEP1isnotSEP} (see Appendix \ref{sec:AppendixSepEnsTransfo} for further details).\\

\section{Approximate transformations}

Having introduced approximate transformations and identified five physically relevant types ($T_{ens}^{\delta,\epsilon}(\psi,\phi)$ etc.), we now set out to better understand these transformations. A complete characterisation of these sets of transformations seems unfeasible. To see this, consider for instance the set $T_{det}^{\delta,\epsilon}(\psi,\phi)$. In order to characterise this set, one must characterise all deterministic transformations from any state in the $\delta$-vicinity around $\ket{\psi}$ to a state in the $\epsilon$-vicinity around $\ket{\phi}$. However, for any finite $\delta>0$, there may be an infinite number of SLOCC classes intersecting the $\delta$-vicinity around $\ket{\psi}$. Moreover, in the general cases, one must also consider transformations to non-fully-entangled states, which is challenging by itself. 

Nonetheless, in the coming sections, we will lay out the landscape of approximate transformations. We have seen that, in the bipartite setting, the hierarchy between these transformations collapses as the optimal fidelity is always achievable with a deterministic transformation. In the following, we will show that the multipartite landscape is considerably richer. In Section \ref{sec:MES}, we begin by considering the maximally entangled set (MES) under approximate transformations. Then in Section \ref{sec:delta0}, we consider transformations in which we start exactly with the initial state. That is, we set $\delta=0$. Here we start in Section \ref{sec:Delta0WithinAnSLOCCclass} with the simplest approximate transformations, i.e., $T_{ens-SLOCC}^{0,\epsilon}$ and $T_{det-SLOCC}^{0,\epsilon}$. We show that, unlike the bipartite case, there are ensemble transformations within an SLOCC class that are better than any deterministic transformation within an SLOCC class (see Fig \ref{fig:SLOCCensBeatsSLOCCdet}).

\begin{figure}[h]
  \includegraphics[width=.95\linewidth]{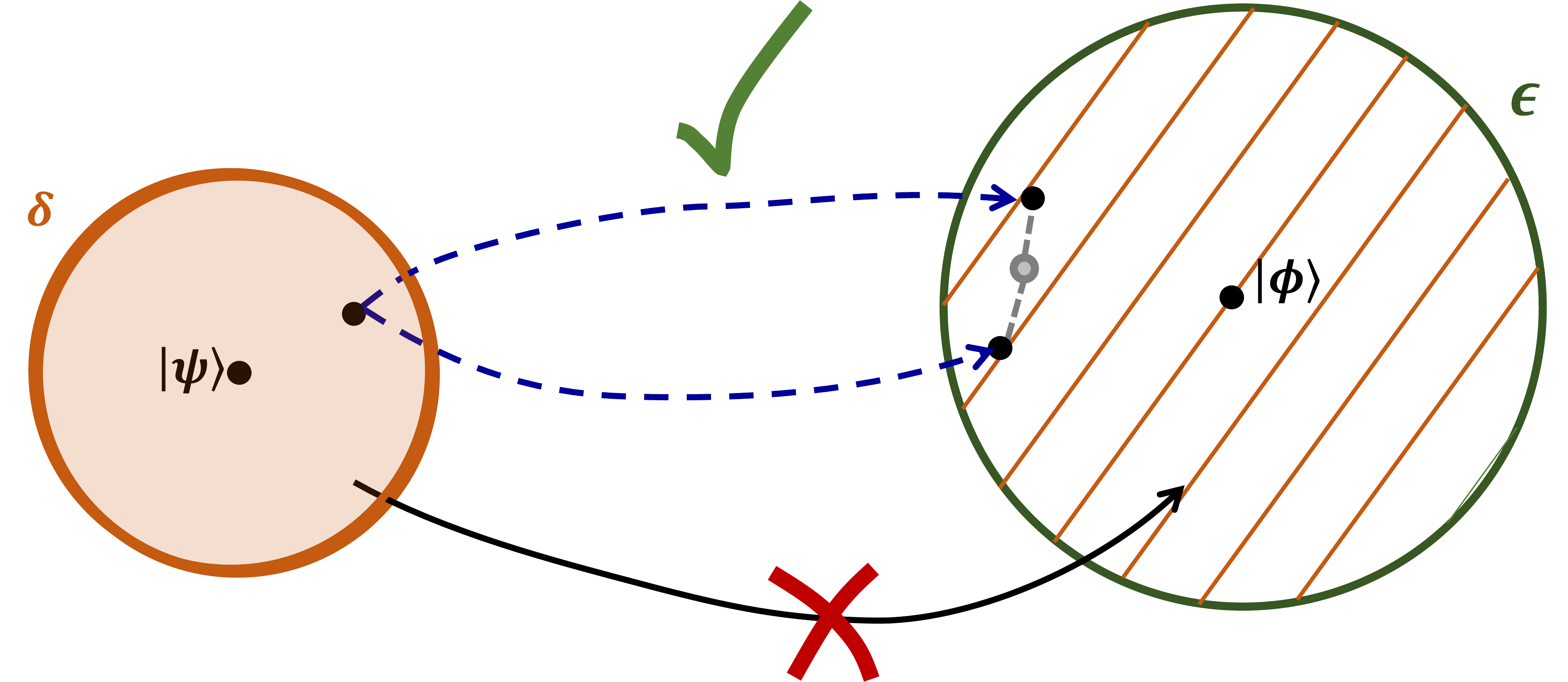}
  \caption{Approximate transformations within an SLOCC class are more powerful than deterministic transformations within an SLOCC class.}
  \label{fig:SLOCCensBeatsSLOCCdet}
\end{figure}

From here, we move on to study general faithful transformations in Section \ref{sec:Delta0General}. Here we study how $T_{ens}^{0,\epsilon},$ and $T_{det}^{0,\epsilon}$ relate to one another. We show that, unlike the bipartite case, there are cases in the multipartite setting where the optimal approximate transformation is non-deterministic (see Fig. \ref{fig:OptIsNotDet}). We also present a multipartite example where the optimal faithful transformation is deterministic and extend these results to the setting where $\delta>0$, but sufficiently small. In Appendix \ref{sec:Appendixlimitingcase}, we also discuss the limiting case where $\epsilon\rightarrow 0$. 

\begin{figure}[h]
  \includegraphics[width=.8\linewidth]{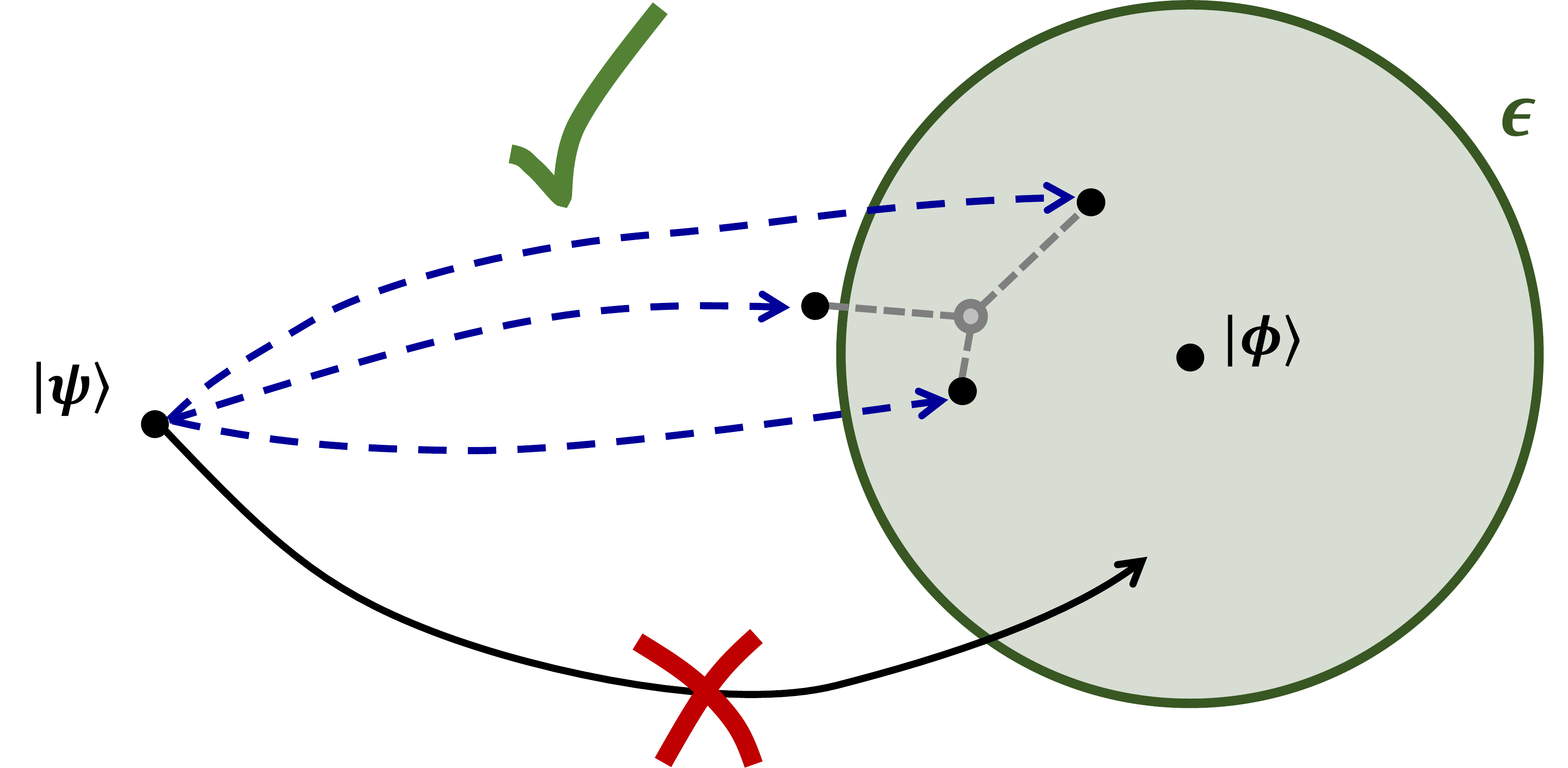}
  \caption{In the multipartite case, the optimal transformation is not always deterministic. There are choices of states $\ket{\psi}$ and $\ket{\phi}$ and $\epsilon>0$ such that a faithful ensemble transformation of $\ket{\psi}$ is possible (blue dashed lines), but there is no deterministic transformation of $\ket{\psi}$ into the $\epsilon$-vicinity around $\ket{\phi}$.}
  \label{fig:OptIsNotDet}
\end{figure}

Finally, in Section \ref{sec:delta=epsilon}, we tackle the question of whether there are really significantly more approximate transformations than deterministic transformations;  namely, are there approximate transformations which are not in the vicinity of any exact, deterministic transformation? More precisely, in Section \ref{sec:delta=epsilon}, we set $\delta=\epsilon>0$ -- in effect, fixing a resolution, up to which we can identify states -- and provide strong numerical evidence for the existence of an approximate transformation with no deterministic transformation between the corresponding $\epsilon$-vicinities around the initial and final states (see Fig. \ref{fig:ApproxIsMoreThanDet}).

\begin{figure}[h]
  \includegraphics[width=.9\linewidth]{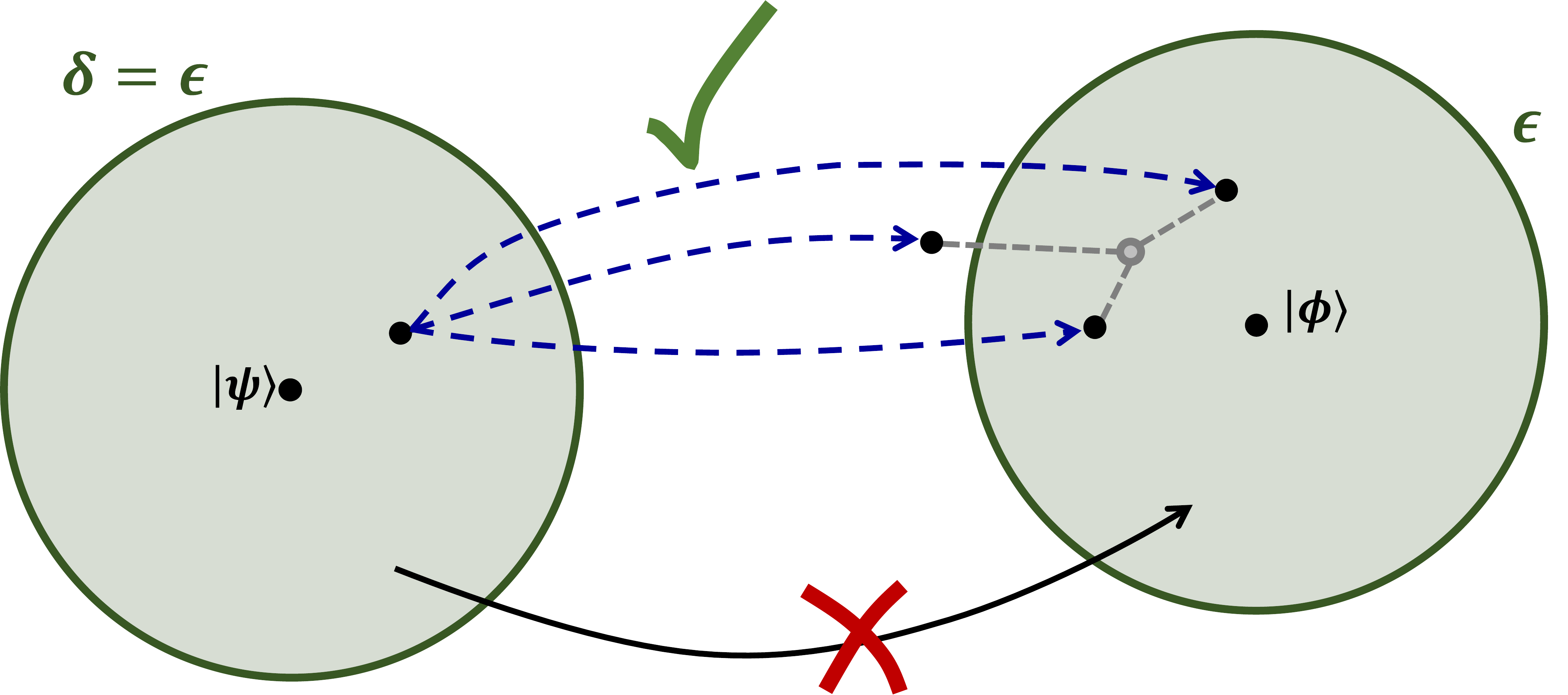}
  \caption{We provide strong numerical evidence for the existence of a faithful transformation for which no deterministic transformations between the $\epsilon$-vicinities of the initial and final states exist. Thus, approximate transformations are more powerful than deterministic transformations.}
  \label{fig:ApproxIsMoreThanDet}
\end{figure}

\subsection{MES}
\label{sec:MES}

We begin by studying the Maximally Entangled Set (MES) under approximate transformations. To begin, contrary to deterministic LOCC in the multipartite setting \cite{GourKrausWallach2017_AlmostAllTrivStab, SauerweinEtAl2018_AlmostAllStatesNotReachable} (see the preliminaries), multipartite states can never be isolated under general approximate transformations for any finite $\delta$ and/or $\epsilon$. The reason for this is trivial: all states are always convertible to states $\epsilon$-close by simply doing nothing. This observation has direct consequences for the MES. Namely, for any finite $\epsilon>0$, we can cover the Hilbert space in finitely many $\epsilon$-vicinities. Therefore, under general approximate transformations, the MES is always finite. However, unlike in the exact case, the MES under faithful transformations is non-unique (for a further discussion of the MES under larger sets of operations than LOCC, see Ref. \cite{NevenEtAl2021_Multicopy}). In the case of SLOCC preserving transformations, by the same argument, we also have that, for any given SLOCC class and any finite $\epsilon$ and/or $\delta>0$, a finite set of states in the SLOCC class are sufficient to reach every state in the SLOCC class via SLOCC-preserving approximate transformations.

\subsection{Transformations from a Fixed Initial State ($\delta=0$)}
\label{sec:delta0}
In this section, we consider the case where we start exactly with $\ket{\psi}$ (i.e., set $\delta=0$) and transform it $\epsilon$-close (with $\epsilon>0$) to some desired state, $\ket{\phi}$. We will show ensemble transformations within an SLOCC class, $T_{ens-SLOCC}^{0,\epsilon}$, are more powerful than deterministic transformations within an SLOCC class, $T_{det-SLOCC}^{0,\epsilon}$. We will also show the optimal transformation is not always deterministic, and thus ensemble faithful transformations, $T_{ens}^{0,\epsilon}$, are in general more powerful than deterministic faithful transformations, $T_{det}^{0,\epsilon}$. Finally, we will provide examples where the optimal transformation is deterministic.

\subsubsection{Transformations within an SLOCC Class}
\label{sec:Delta0WithinAnSLOCCclass}

We begin our investigation by studying ensemble transformations within an SLOCC class and comparing $T_{ens-SLOCC}^{0,\epsilon}(\psi,\phi)$ and $T_{det-SLOCC}^{0,\epsilon}(\psi,\phi)$. Throughout this subsection, we consider $\ket{\psi}$ and $\ket{\phi}$ to be SLOCC equivalent. Recall in the bipartite case, we had that the optimal fidelity under $T_{ens-SLOCC}$ is equal to the optimal fidelity under $T_{det-SLOCC}$. Moreover, the optimal fidelity is always achievable, and thus $T_{ens-SLOCC}^{0,\epsilon}(\psi, \phi)$ is empty iff $T_{det-SLOCC}^{0,\epsilon}(\psi, \phi)$ is empty in the bipartite case.

In this section, we show that this result does not generally hold in the multipartite case. More precisely, we show that there exist SLOCC equivalent states, $\ \ket{\psi}\cong_{SLOCC} \ket{\phi}$, and an $\epsilon > 0$ such that
\begin{align}
    T_{ens-SLOCC}^{0,\epsilon}(\psi, \phi)\supsetneq T_{det-SLOCC}^{0,\epsilon}(\psi, \phi)=\emptyset 
    \label{eq:EnsSLOCCbetterthanDetSLOCC}
\end{align}
and
\begin{align}
    F_{ens-SLOCC}^0(\psi \rightarrow \phi)>F_{det-SLOCC}^0(\psi \rightarrow \phi).
\end{align}
We will show this by constructing an ensemble transformation that is \textit{not} SLOCC-preserving, i.e., $\Lambda\in T_{ens}(\psi)\setminus T_{ens-SLOCC}(\psi)$, but has a considerably higher fidelity than $F_{det-SLOCC}^0(\psi\rightarrow\phi)$. We will then use the fact that transformations in $T_{ens-SLOCC}(\psi)$ approximate transformations in $T_{ens}(\psi)$ arbitrarily well (see Theorem \ref{thm:Fens=FensSLOCC}) to prove that there is a transformation in $T_{ens-SLOCC}(\psi)$ which is better than any transformation in $T_{det-SLOCC}(\psi)$.

To construct our example, we consider the following initial state:
\begin{equation}
    \ket{\psi}\propto \sqrt{7}\ket{00000}+ \sqrt{5} \ket{11111} + \sqrt{10} \ket{D_{5,3}},
    \label{eq:isolatedstate}
\end{equation}
where $\ket{D_{5,3}}$ is the Dicke state of five qubits with three excitations. This state is generic in the sense described in the preliminaries \cite{SauerweinEtAl2018_AlmostAllStatesNotReachable}; in particular, it is isolated. Moreover, it is critical and permutationally symmetric, which will make calculations simpler.

As a target state, we consider a 1-parameter family of states given by
\begin{equation}
    \ket{\phi(\lambda)} \propto D^{\otimes 5}_\lambda \ket{\psi},
\end{equation}
where $D_\lambda=\text{diag}(1/2+\lambda ,1/2-\lambda)$ and $\lambda \in (0,1/2)$. Note, in the limit $\lambda\rightarrow1/2$, the local operator $D_\lambda$ converges to the projector $|0\rangle \langle 0|$. As $\langle0|^{\otimes 5}\ket{\phi(\lambda)}\ne 0$,  $\ket{\phi(\lambda)}$ therefore converges to $\ket{0}^{\otimes 5}$.

We begin with $T_{det-SLOCC}^{0,\epsilon}(\psi,\phi)$. As the initial state is generic and therefore isolated, the only fully-entangled states it can deterministically reach with LOCC are LU equivalent states. Therefore,
$F_{det-SLOCC}^0(\psi\rightarrow \phi(\lambda))=F_{LU}(\psi,\phi(\lambda))$. Moreover, we have
\begin{equation}
    T_{det-SLOCC}^{0,\epsilon}(\psi,\phi(\lambda))\ne \emptyset \text{ iff } F_{LU}(\psi,\phi(\lambda))\ge 1-\epsilon.
    \label{eq:necconditiongenericnotempty}
\end{equation}

Note that, even though both $\ket{\psi}$ and $\ket{\phi(\lambda)}$ are permutationally-invariant states with positive coefficients in the computational basis \cite{Hubeneretal2009_GeometricEntSymStates, NevenEtAl2016_FidelitySymStates, AulbachEtAl2010_GeomEntSymPosStates}, an analytic expression for $F_{LU}(\psi,\phi(\lambda))$ is still not easily attainable. As such, we will often consider bipartite unitaries to determine upper bounds. Namely, $F_{LU}(\psi,\phi(\lambda))=\max_{\otimes U_i} F(\ket{\psi},\otimes_i U_i \ket{\phi})\le \max_{U,V} F(\ket{\psi}, U\otimes V \ket{\phi})$. This last expression, can be easily evaluated \cite{VidalJonathanNielsen2000_ApproxLOCC}. For instance, we have that
\begin{equation}
    F_{LU}(\psi,\phi(\lambda))\le \frac{1}{2} \left(\sum_i \sqrt{\mu_i(\rho_1(\lambda))}\right)^2,
\end{equation}
where $\mu_i(\rho_1(\lambda))$ are the eigenvalues of the reduced density matrix on qubit 1 for $\ket{\phi(\lambda)}$. We will use this upper bound frequently through-out this paper \footnote{Note, we could use the reduced density matrix for qubits 1 and 2 to obtain a tighter bound, but it makes the expression more complicated and doesn't change any of the subsequent analysis.}. 

As a general ensemble transformation, we consider the standard One Successful Branch Protocol (OSBP). In this protocol, all parties measure simultaneously. If they are all successful, they will have transformed $\ket{\psi}$ to $\ket{\phi(\lambda)}$. This successful outcome occurs with maximum probability, $p_{max}(\ket{\psi}\rightarrow \ket{\phi(\lambda)})\equiv p_{max}(\lambda)= \frac{n_\lambda^2}{(1+\lambda)^{10}}$, where $n_\lambda^2=||D_\lambda^{\otimes5}\ket{\psi}||^2$  \cite{GourKrausWallach2017_AlmostAllTrivStab,SauerweinEtAl2018_AlmostAllStatesNotReachable} (see Eq. (\ref{eq:pmaxgeneric})). In the event, at least one party does not get a successful outcome, then all parties convert the post-measurement state to a product state. For our protocol, we choose this product state to be $\ket{00000}$. We choose this because (as is easy to verify using the fact that $\ket{\phi(\lambda)}$ is permutationally symmetric and only has positive coefficients \cite{Hubeneretal2009_GeometricEntSymStates,NevenEtAl2016_FidelitySymStates, AulbachEtAl2010_GeomEntSymPosStates}), the product state nearest to $\ket{\phi(\lambda)}$ is $\ket{00000}$, $\forall \lambda \in (0.00416,1/2)$, i.e.,
\begin{equation}
    \ket{00000}= \argmax_{\otimes_i \ket{\tilde{e}_i}} F\left( \otimes_i \ket{\tilde{e}_i} , \ket{\phi(\lambda)}\right).
    \label{eq::SloccPreservingNearestProductState}
\end{equation}
$\forall \lambda \in (0.00416,1/2)$ \footnote{As in all the following arguments we end up considering $\lambda>0.0696$, we simplify our analysis by choosing $\ket{00000}$ as the failing-outcome product state, $\forall \lambda\in(0,1/2)$, rather than modifying the protocol for the case of $\lambda<0.00416$}. The average fidelity of this protocol is then given by
\begin{equation}
    F_{av}(\Lambda_{OSBP}^\lambda(\psi), \phi(\lambda)) = p_{max}(\lambda)+(1-p_{max}(\lambda)) F_0(\lambda),
\end{equation}
where $F_0(\lambda)= F(\ket{0}^{\otimes5},\ket{\phi(\lambda)})$. 

It is easy to verify that, for $\lambda>\lambda_0=0.0696$, $F_{av}(\Lambda_{OSBP}^\lambda(\psi), \phi(\lambda))$ is strictly larger than $F_{LU}(\psi,\phi(\lambda))$. As, by Theorem \ref{thm:Fens=FensSLOCC}, we have that there is an SLOCC-preserving ensemble transformation with an average fidelity arbitrarily close to $F_{av}(\Lambda_{OSBP}^\lambda(\psi), \phi(\lambda))$, we can therefore conclude that there is an SLOCC-preserving ensemble transformation with a strictly better average fidelity than $F_{det-SLOCC}^0(\psi\rightarrow \phi(\lambda))$. That is, for $\lambda\in(\lambda_0,1/2)$
\begin{equation}
    F_{ens-SLOCC}^0(\psi\rightarrow \phi(\lambda))>F_{det-SLOCC}^0(\psi\rightarrow \phi(\lambda)).
\end{equation}
Moreover, for $\lambda\in(\lambda_0,1/2)$, we also have that there exists an $\epsilon>0$ such that
\begin{align}
    T_{ens-SLOCC}^{0,\epsilon}(\psi, \phi)\supsetneq T_{det-SLOCC}^{0,\epsilon}(\psi, \phi)=\emptyset.
    \label{eq:TensSLOCCMoreThanTdetSLOCC}
\end{align}

A few comments are now in order. First, it is easy to see that this argument extends to any generic initial state and any SLOCC equivalent target state sufficiently far away. That is, for almost all states (generic states are full-measure), $\ket{\psi}$, and any SLOCC equivalent state, $\ket{\phi}$, such that $F_{LU}(\psi,\phi)<F_{av}(\Lambda_{OSBP}(\psi),\phi)$, there is an $\epsilon>0$ such that Eq. (\ref{eq:TensSLOCCMoreThanTdetSLOCC}) holds. Second, this argument can easily be extended to sufficiently small $\delta>0$. This holds as the set of generic states is open \cite{GourKrausWallach2017_AlmostAllTrivStab, SauerweinEtAl2018_AlmostAllStatesNotReachable} \footnote{Note, normalised generic states are also open and dense in the set of normalised states with respect to the topology induced by the trace distance and are full measure, see \cite{Slowik2020_SLOCCtypes}.}, and therefore, for sufficiently small $\delta$, all states in the $\delta$-vicinity are also generic. On must also choose $\ket{\phi}$ sufficiently far away that the $\delta$ and $\epsilon$-vicinities do not overlap (see Appendix \ref{sec:apppendixoverlapbetweenballs}). Finally, as the results of Ref. \cite{SauerweinEtAl2018_AlmostAllStatesNotReachable} apply for $n\ge5$, these results also hold for larger system sizes.

\subsubsection{Comparison between General Transformations}
\label{sec:Delta0General}

In this section, we continue to consider transformations with $\delta=0$ but now consider general transformations. Specifically, we study how $T_{ens}^{0,\epsilon}(\psi,\phi)$ and $T_{det}^{0,\epsilon}(\psi,\phi)$ relate to one another. Recall that, in the bipartite setting, we have that the optimal fidelity is always achievable with a deterministic transformation and $T_{ens}^{0,\epsilon}(\psi,\phi)$ is empty iff $T_{det}^{0,\epsilon}(\psi,\phi)$ is empty. In this section, we show that the multipartite setting is considerably more nuanced. In particular, we show that, unlike the bipartite case, the optimal fidelity cannot generally be achieved with a deterministic transformation. More precisely, we show that there exist states $ \ket{\psi}$ and $\ket{\phi}$ such that
\begin{equation}
    F_{ens}^{0}(\psi\rightarrow \phi)>F_{det}^{0}(\psi\rightarrow \phi),
\end{equation}
and that, for these states, there exist $\epsilon>0$ such that
\begin{equation}
    T_{ens}^{0,\epsilon}(\psi,\phi)\supsetneq T_{det}^{0,\epsilon}(\psi,\phi)=\emptyset.
    \label{eq:faithfulmorethandeterministic}
\end{equation}
We also show, unsurprisingly, that there exist pairs of states where the optimal transformation is deterministic.

So, as outlined, we begin by proving that, contrary to the bipartite setting, the optimal fidelity cannot generally be achieved with a deterministic transformation. To see this, consider two five-qubit states $\ket{\psi}$ and $\ket{\phi}$, such that a deterministic transformation from $\ket{\psi}$ to $\ket{\phi}$ is possible, but $\ket{\psi}$ is not close to $\ket{\phi}$. Now let us pick a generic state, $|\tilde{\psi}_\epsilon \rangle$, $\delta=\epsilon^2$-close to $\ket{\psi}$. It follows from the inequalities Eq. (\ref{eq:vidalinequality}) \footnote{It is easy to verify that Eq. (\ref{eq:vidalinequality}) also holds in the multipartite case.} and Eq. (\ref{eq:FidInequality}) that $F(|\tilde{\psi}_\epsilon\rangle \rightarrow \ket{\phi})\ge 1-\epsilon$; that is, we can faithfully transform $|\tilde{\psi}_\epsilon \rangle$ to $\ket{\phi}$. Note that, as the set of generic states is dense \cite{GourKrausWallach2017_AlmostAllTrivStab, SauerweinEtAl2018_AlmostAllStatesNotReachable}, $\epsilon$ can be chosen arbitrarily small. We now show that we can choose $\epsilon$ such that there is no deterministic transformation from  $|\tilde{\psi}_\epsilon \rangle$ to $\ket{\phi}$. As $|\tilde{\psi}_\epsilon \rangle$ is generic, it can only be transformed into either (a) a non-fully-entangled state or (b) an LU equivalent state. Regarding (a), we can choose $\epsilon$ sufficiently small that the $\epsilon$-vicinity around $\ket{\phi}$ only contains fully-entangled states (we can do this as the set of separable states is closed). Regarding (b), it is easy to verify, using Eq. (\ref{eq:FidInequality}) and the metric properties of the trace distance, that provided $\max_{\otimes U_i} F(\ket{\psi},\otimes_i U_i \ket{\phi})\ll 1-(\sqrt{\epsilon}+\epsilon)^2$, $|\tilde{\psi}_\epsilon \rangle$ cannot be transformed into the $\epsilon$-vicinity around $\ket{\phi}$ with LUs. Thus, we have constructed a faithful transformation from $|\tilde{\psi}_\epsilon \rangle$ to $\ket{\phi}$ such that there is no deterministic transformation from $|\tilde{\psi}_\epsilon \rangle$ to the $\epsilon$-vicinity around $\ket{\phi}$. Consequently, Eq. (\ref{eq:faithfulmorethandeterministic}) holds and the optimal faithful transformation is strictly not deterministic.

As a concrete example, consider the transformation $\ket{\psi}=\ket{GHZ_5}\propto\ket{00000}+\ket{11111}$ to $\ket{\phi}\propto \left(\mathbf{1}^{\otimes 4}\otimes\text{diag}(\frac{2}{3},\frac{3}{2})\right)\ket{GHZ_5}$. Such a transformation is possible via LOCC \cite{Turgut2010_GHZtransfo}. By considering all possible bipartite splittings, we find that the closest biseparable, pure state has fidelity $\frac{81}{97}\approx 0.835$ with $\ket{\phi}$. Moreover, the maximum fidelity over bipartite unitaries upper bounds the maximum fidelity over local unitaries. Thus, from Ref. \cite{VidalJonathanNielsen2000_ApproxLOCC}, we have $\max_{U_i} F( \otimes U_i|\psi \rangle,\ket{\phi}) < \frac{169}{194}\approx0.871$. Now let $\epsilon< \frac{1-\sqrt{1+4\sqrt{1-\frac{169}{194}}}+2\sqrt{1-\frac{169}{194}}}{2} \approx0.0786$ and consider a generic state, $|\tilde{\psi}\rangle$ such that $F(\ket{GHZ_5},|\tilde{\psi}\rangle)>1-\epsilon^2$ (which, as previously argued, is guaranteed to exist). Then by the above argument, $|\tilde{\psi}\rangle$ can be transformed $\epsilon$-close to $\ket{\phi}$ with an ensemble transformation, whereas no deterministic transformation can achieve this accuracy.  

 Finally, note that this argument easily extends to the case of sufficiently small $\delta>0$ (chosen independently from $\epsilon$). As in the previous section, this is true because the set of generic states used in the above argument is open \cite{GourKrausWallach2017_AlmostAllTrivStab, SauerweinEtAl2018_AlmostAllStatesNotReachable} (see the preliminaries). Consequently, for any generic state, $|\tilde{\psi}\rangle$, there is a $\delta>0$ such that the $\delta$-vicinity around $|\tilde{\psi}\rangle$ consists of only generic states. Therefore (assuming $\delta$ and $\epsilon$ vicinities do not overlap [see Appendix \ref{sec:apppendixoverlapbetweenballs}]), the entire argument above goes through. We emphasise that this only holds for small enough $\delta$, as the argument depends on all states in the $\delta$-vicinity being generic; for example, in our construction above, the deterministically transformable state $\ket{\psi}$ was $\epsilon^2$-close to $|\tilde{\psi}\rangle$. Therefore, in order for the $\delta$-vicinity around $|\tilde{\psi}\rangle$ to consist only of generic states, $\delta$ must be necessarily smaller than $\epsilon^2$, i.e., $\delta<\epsilon^2\ll \epsilon$ (indeed, it most likely will have to be chosen yet smaller). Lastly, as the above argument follows from the existence of a dense, open set of isolated states, the conclusions also hold true for $n\ge 5$ and qudits \footnote{In the case, $d\ge3$, the result holds for $n\ge 4$} \cite{GourKrausWallach2017_AlmostAllTrivStab,  SauerweinEtAl2018_AlmostAllStatesNotReachable}.  

Having demonstrated that, in general, optimal transformations are not deterministic, we now round out this insight and demonstrate that, whilst in general deterministic transformations are sub-optimal, there are cases where the optimal approximate transformation is deterministic. That is, we show that, for all system sizes, there is a $\ket{\psi}$ and a $\ket{\phi}$ such that $\ket{\psi}\not\rightarrow_{LOCC}\ket{\phi}$, yet $F^0_{ens}(\ket{\psi}\rightarrow\ket{\phi}) = F^0_{det}(\ket{\psi}\rightarrow\ket{\phi})$.

To see this, we build a construction using the arguments from Ref. \cite{VidalJonathanNielsen2000_ApproxLOCC}. Consider as an initial state the following states in the GHZ class, $\ket{\psi_\lambda}=\sqrt{d}\ (\sqrt{D_\lambda}\otimes \mathbbm{1}^{\otimes n-1}) \ket{GHZ_n^d}$, where $\ket{GHZ_n^d}=\frac{1}{\sqrt{d}}\sum_{i=0}^{d-1} \ket{ii...i}$ is the n-qudit GHZ state and $\sqrt{D_\lambda}=\text{diag}(\sqrt{\lambda_0},\sqrt{\lambda_1},...\sqrt{\lambda_{d-1}})$ is a diagonal positive matrix such that $\sum_i \lambda_i =1$. As a target state, we consider the n-qudit GHZ state itself, i.e., $\ket{\phi}=\ket{GHZ_n^d}$. Note that if $\lambda_i \ne 1/d$  for any $i\in\{0,..,d-1\}$, then $\ket{\psi_\lambda}\not\rightarrow_{LOCC} \ket{\phi}$. The fact that the optimal faithful transformation is deterministic then follows from a similar argument as in Ref. \cite{VidalJonathanNielsen2000_ApproxLOCC}. Namely, it is easy to see that the optimal bipartite LOCC transformation is to do nothing, which naturally is achievable in the multipartite setting too.

 As a final comment, one can consider the case of $\delta>0$. Here we do not provide a proof, but Ref \cite{Acin2000_OptimalDistillGHZ} reported strong numerical evidence that, in the case of three qubits, the optimal protocol for transforming any GHZ-like state to the GHZ state is to apply LUs.

\subsection{Transformations with respect to a Fixed Resolution, ($\delta=\epsilon>0$)}
\label{sec:delta=epsilon}

In the previous section, we considered faithful transformations that start precisely with the initial state ($\delta=0)$. Moreover, we saw that whenever our results hold due to the initial state being a generic state, then our results also hold for sufficiently small $\delta$. This is because the set of generic states is open, and so, for sufficiently small $\delta$, the $\delta$-vicinity around a generic state consists only of generic states. This means, for instance, that there are $\ket{\psi}$, $\ket{\phi}$ and $\delta>0$ such that
\begin{equation}
    F_{ens}^{\delta}(\psi,\phi)>F_{det}^{\delta}(\psi,\phi).
    \label{eq:deltaapproxmorethanlocc}
\end{equation}

However, the constructions in the previous sections were nevertheless ``near''  to a known deterministic transformation. For example, we demonstrated Eq. (\ref{eq:deltaapproxmorethanlocc}) by considering a generic state near the initial state of a known deterministic transformation and using the inequality in Eq. (\ref{eq:vidalinequality}). This begs the question of whether approximate transformations are more than just deterministic transformations between the $\epsilon$-vicinities around the initial and final state. It is this question that we tackle in this section.

To this end, we consider the case where $\delta=\epsilon$. Physically, this is motivated by imagining a scenario in which we have some fixed resolution, $\epsilon>0$. In this case, both the initial and final state can only be known up to a resolution of $\epsilon$. Thus, we wish to know if any of the pure states that could be misidentified as the initial state can be deterministically transformed into a pure state that could be misidentified as the target state.  Note that it is not straightforward to simply use the inequality in Eq. (\ref{eq:vidalinequality}) to construct an approximate transformation that is not near a deterministic transformation (as done in Section \ref{sec:Delta0General}). In fact, it is easy to verify that setting $\delta=\epsilon$ rules out such a construction.

We note that in the bipartite case, we still have $T_{ens}^{\epsilon,\epsilon}(\psi,\phi)=\emptyset$ iff $T_{det}^{\epsilon,\epsilon}(\psi,\phi)=\emptyset$. In the remainder of this section, we provide strong numerical evidence that approximate transformations are more powerful than deterministic transformations between the $\epsilon$-vicinities around $\ket{\psi}$ and $\ket{\phi}$. More precisely, we provide strong numerical evidence that there is a $\ket{\psi}$, a $\ket{\phi}$ and an $\epsilon>0$ such that
\begin{equation}
    T_{ens}^{\epsilon,\epsilon}(\psi, \phi) \supsetneq T_{det}^{\epsilon,\epsilon}(\psi, \phi) = \emptyset,
    \label{eq:approxmorethandeterministic}
\end{equation}
and thus, in the multipartite case, approximate transformations really are more than deterministic transformations between the $\epsilon$-vicinities around the initial and final state.

The idea of our construction is to consider generic states and have $\epsilon$ small enough such that the $\epsilon$-vicinity around the initial state consists only of generic states. Consequently, our protocol will have to have an average fidelity of almost 1 as, in order for $\Lambda\in T_{ens}^{\epsilon,\epsilon}(\psi,\phi)$, we must have $F_{av}(\Lambda(\psi),\phi)\ge 1-\epsilon$. This will ultimately lead us to considering transformations to target states that, though still entangled, are very close to $\ket{0}^{\otimes 5}$. Nonetheless, by choosing $\epsilon$ this way, we simplify our problem considerably as all states in the initial $\epsilon$-vicinity can only be transformed into either (a) LU equivalent states or (b) non-fully-entangled states.  

We introduce a family of protocols, $\Lambda_\lambda$ with $\lambda\in(0,1/2)$, which transform an initial, generic state, $\ket{\psi}$, towards a generic, target state, $\ket{\phi(\lambda)}$. We find in the limiting case $\lim_{\lambda\rightarrow 1/2}F_{av}(\Lambda_\lambda(\psi),\phi(\lambda))=1$. Therefore, by a suitable choice of $\lambda\in(0,1/2)$, our protocol does indeed achieve an average fidelity $F_{av}(\Lambda_\lambda(\psi),\phi(\lambda))\ge 1-\epsilon$, with $\epsilon$ sufficiently small to ensure that all states in the $\epsilon$-vicinity around the initial state are generic states. We then upper bound the maximum achievable fidelity of deterministic transformations from the $\epsilon$-vicinity around the initial state. That is, we upper bound  $F_{det}^\epsilon(\psi,\phi(\lambda))$. Finally, we put these results together, showing strong numerical evidence that there is a $\lambda \in(0,1/2)$ and $\epsilon>0$ such that  $F_{av}(\Lambda_\lambda(\psi),\phi(\lambda))\ge 1-\epsilon > F_{det}^\epsilon(\psi,\phi(\lambda))$ and thus there is a $\ket{\psi}$, $\ket{\phi}$ and $\epsilon>0$ such that
\begin{equation}
    T_{ens}^{\epsilon,\epsilon}(\psi, \phi) \supsetneq T_{det}^{\epsilon,\epsilon}(\psi, \phi) = \emptyset.
\end{equation}

Our construction is a modification of the one we introduced in Section \ref{sec:Delta0WithinAnSLOCCclass}. Once again, we consider the generic, permutationally invariant state in Eq. (\ref{eq:isolatedstate}), $\ket{\psi}$, as the initial state and  consider the family of states $ \ket{\phi (\lambda)} \propto D_\lambda^{\otimes 5} \ket{\psi}$, with $D_\lambda = \text{diag}\left(1/2+\lambda, 1/2-\lambda \right)$ and $\lambda\in (0,1/2)$, as target states. As a protocol, we consider a modification of the OSBP from Section \ref{sec:Delta0WithinAnSLOCCclass} \footnote{Note that neither the OSBP nor the natural extension of it involving a single four qubit state leads to a higher fidelity as the protocol presented here.}. Recall, in the original protocol, all parties measure simultaneously; if all the parties get a successful outcome, then they will have transformed the initial state to the target state (this occurs with maximal probability); but if one of the parties gets a non-successful outcome, then they all transform the post-measurement state to the nearest product state. 

In our modified protocol, the parties do not perform their measurements simultaneously but instead perform them sequentially (see Fig. 5). Each party performs in turn the measurement:
\begin{align}
    M_0 &= \frac{1}{1/2+\lambda} D_\lambda\\
    M_1 &=  \sqrt{\mathbbm{1}-M_0^\dagger M_0} \propto \ket{1}\bra{1}.
\end{align}

This is the measurement that would be used in an OSBP. In the event each party measures $M_0$ (which happens for each party with probability $p_i$), they obtain the target state. This outcome occurs with probability $\Pi_i p_i (\lambda)= p_{max}(\ket{\psi}\rightarrow_{LOCC} \ket{\phi(\lambda)})\equiv p_{max}(\lambda)$ \cite{GourKrausWallach2017_AlmostAllTrivStab,SauerweinEtAl2018_AlmostAllStatesNotReachable}. Thus indeed, this protocol will be an optimal conclusive transformation (i.e., belongs to $T_{pmax}(\psi,\phi(\lambda))$).

However, unlike in the standard OSBP, in the event one party in the sequence does not get the successful outcome, the parties do not simply transform to a product state. Let party $k\in\{1,..,5\}$ measure outcome $M_1$. In this case, the post-measurement state will be
\begin{equation}
    |\chi_k(\lambda)\rangle \propto D_\lambda^{\otimes k-1}\otimes \mathbbm{1}^{\otimes 6-k} \ket{\tilde{\chi}_k}, 
\end{equation}
with
\begin{equation}
    \ket{\tilde{\chi}_k} = \ket{1}_k \otimes \left( \sqrt{\frac{5}{11}}\ket{1111} + \sqrt{\frac{6}{11}} \ket{D_{4,2}}\right)_{\{1,...,5\}\setminus\{k\}}.
\end{equation}

$|\chi_k(\lambda)\rangle$ is no longer fully-entangled, and thus there is zero probability of transforming it to $\ket{\phi(\lambda)}$. Instead of simply transforming this failing branch into a product state, the parties apply a conclusive transformation to a four qubit state that optimises the average output fidelity.

To be more precise, first note that for any given state, $|\tilde{\xi}\rangle\cong_{SLOCC} \ket{\chi_k(\lambda)}$, we can analytically lower bound the maximum probability of transforming $\ket{\chi_k(\lambda)}$ to $|\tilde{\xi}\rangle$ by \cite{GourKrausWallach2017_AlmostAllTrivStab, SauerweinEtAl2018_AlmostAllStatesNotReachable}
\begin{equation}
    p_{max}\left(|\chi_k(\lambda)\rangle\rightarrow  |\tilde{\xi}\rangle \right)\ge \frac{\left|\left|\otimes g_i \ket{\chi_k(\lambda)} \right|\right|^2}{\Pi_i \mu_{max}(g_i^\dagger g_i)} \equiv q_{\lambda,k}(\tilde{\xi}).
    \label{eq:FourQubitProb}
\end{equation}
Here, $\otimes_i g_i  \ket{\chi_k(\lambda)} =  |\tilde{\xi}\rangle$ and $\mu_{max} (X)$ is the largest eigenvalue of $X$. 

Returning to our protocol, in the event party $k$ does not get a successful outcome, the parties instead implement the ensemble transformation
\begin{align}
    \ket{\chi_k(\lambda)} \rightarrow& \Big\{ \big(q_k(\lambda), \ket{\xi_k (\lambda)}\big),\big(( 1- q_k(\lambda)), \ket{0}^{\otimes 5}\big)\Big\},
\end{align}
where $q_k(\lambda)\equiv q_{\lambda, k}(\xi_k(\lambda))$, and $\ket{\xi_k (\lambda)}$ is defined such that the average fidelity of this protocol is optimised, i.e.,
\begin{align}
    \ket{\xi_k (\lambda)} \equiv \argmax_{|\tilde{\xi}\rangle\cong_{SLOCC} \ket{\chi_k(\lambda)}} & \Big(q_{\lambda,k}(\tilde{\xi})\ F(|\tilde{\xi}\rangle, \ket{\phi(\lambda)})\ \nonumber\\
    &\qquad + \big(1-q_{\lambda,k}(\tilde{\xi})\big) F_0(\lambda) \Big),
    \label{eq:optimization4qubits}
\end{align}
with $F_0(\lambda)$ being the overlap of $\ket{\phi(\lambda)}$ and $\ket{0}^{\otimes 5}$ (recall $\ket{0}^{\otimes 5}$ is the nearest product state for $\lambda\in(0.00416,1/2)$).

Thus, in the event party $k$ has a non-successful measurement, i.e., measures $M_1$, the parties coordinate to transform $|\chi_k(\lambda)\rangle$ to $\ket{\xi_k(\lambda)}$ with probability $q_k(\lambda)$. If they fail to transform the state into $\ket{\xi_k(\lambda)}$, all parties determinstically transform the post-measurement state to $\ket{0}^{\otimes 5}$.

Putting this altogether, we have an optimal conclusive protocol in which: (a) with probability $\Pi_i p_i(\lambda) = p_{max}(\lambda)$, we get exactly the target state $\ket{\phi(\lambda)}$; (b) with probability
\begin{equation}
    \tilde{q}_k(\lambda)= \left( \Pi_{j=0}^{k-1} p_j(\lambda) \right)\left(1-p_k(\lambda) \right) q_k(\lambda)
\end{equation}
we reach the state $\ket{\xi_k(\lambda)}$, which has fidelity $F_{k}(\lambda)\equiv F(\ket{\xi_k(\lambda)},\ket{\phi(\lambda)})$ with the target state,
and (c) in the remaining branches, the protocol outputs $\ket{0}^{\otimes 5}$ (see Fig. \ref{fig:Protocol}). Thus, the average fidelity of the protocol is given by:
\begin{align}
   F_{prot}(\lambda) &= p_{max}(\lambda) + \sum_{k=1}^{5} \tilde{q}_k(\lambda) F_{k}(\lambda) \nonumber \\
   &\qquad + \Big(1-p_{max}(\lambda)- \sum_{k=1}^{5} \tilde{q}_k(\lambda)\Big)  F_0(\lambda).
\end{align}
 
As analytically evaluating the optimization in Eq. (\ref{eq:optimization4qubits}) is not straightforward \footnote{See Ref. \cite{NevenEtAl2016_FidelitySymStates} for a discussion of optimising the fidelity between permutationally symmetric states over general invertible (as opposed to unitary) operators.}, we numerically optimise to find $|\xi_k(\lambda) \rangle$. As a result, the remainder of this proof remains numerical in nature.

 \begin{figure}
     \centering
     \includegraphics[width=.65\linewidth]{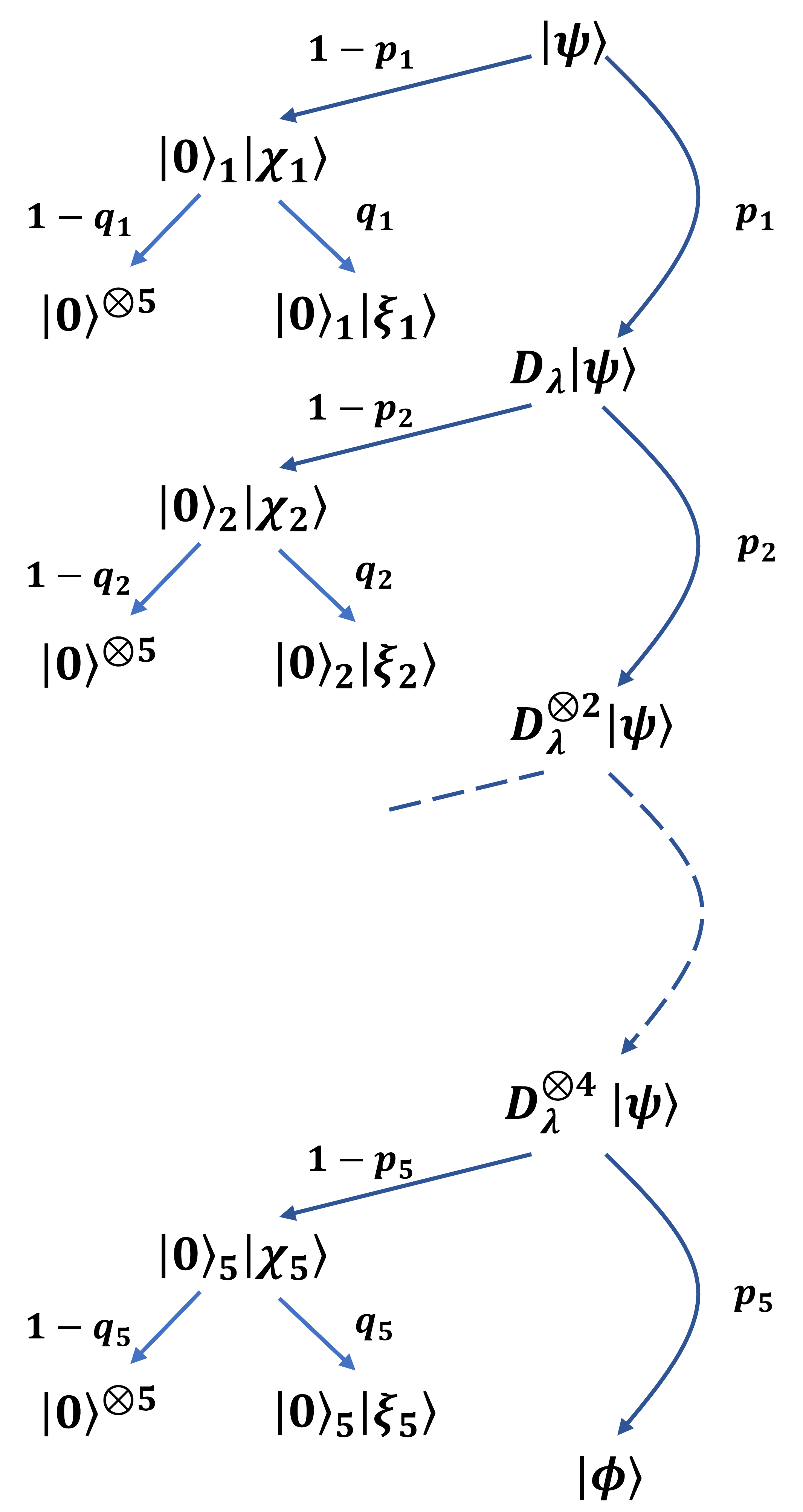}
     \caption{The new protocol proposed in this section. Instead of a standard OSBP, the parties sequentially measure. In the event, one fails, yielding a four qubit state, the parties transform the state probabilistically with a standard OSBP to a four qubit state such that the output fidelity is optimised. In case of an unsuccessful measurement, yielding a three qubit state, the state is transformed deterministically to a product state.}
     \label{fig:Protocol}
 \end{figure}

A numerically derived $F_{prot}(\lambda)$ is plotted in Fig. \ref{Fig:ProtocolFidelity} (the red line). Recall that in order for all states in the initial vicinity to be generic, we require $\epsilon$ to be very small and thus the protocol fidelity to be very large. Considering Fig. \ref{Fig:ProtocolFidelity}, we therefore see that we have two limiting cases of interest:  $\lim_{\lambda\rightarrow 0} F_{prot}(\lambda) = 1$ and $\lim_{\lambda\rightarrow 1/2} F_{prot}(\lambda) = 1$.  We can rule out the first of these limits by the following argument. If the overlap between the two states is greater than $1-\epsilon$, then a deterministic transformation is trivially possible (by simply doing nothing). Indeed, if $\lambda=0$, then one can obviously do the transformation deterministically as the initial and final state are the same. Thus, our protocol must have a greater fidelity than this overlap. The blue line in Fig. \ref{Fig:ProtocolFidelity} shows the overlap between the initial and target state. We can see that $F_{prot}(\lambda)$ is less than this overlap for $\lambda<1/10$. This rules out the limiting case $\lambda \rightarrow 0$. Therefore, we have to consider the limiting case $\lambda \rightarrow 1/2$. We see that for all $\epsilon>0$, there is a $\lambda\in(0,1/2)$ such that $F_{prot}(\lambda)\ge 1-\epsilon$. Consequently, we can choose $\lambda\in (0,1/2)$ such that all states in the initial vicinity are generic states. Next, we upper bound the achievable fidelity of a deterministic transformation, $F_{det}^{\epsilon}(\psi,\phi(\lambda))$. As they are all generic, states in the initial $\epsilon$-vicinity can only be transformed into either (a) LU equivalent states or (b) non-fully-entangled states. Therefore, it remains to show that our protocol has a higher average fidelity than both of these options.

\begin{figure}
\includegraphics[width=0.45\textwidth]{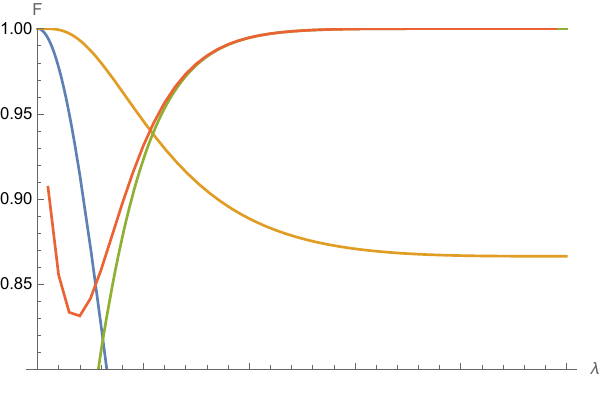}
\centering
\caption{Plot of the fidelities relevant to the proof that approximate transformations are more than deterministic transformations between $\epsilon$-vicinities. The red line corresponds to $F_{prot}(\lambda)$, the blue line corresponds to $|\langle \psi|\phi(\lambda)\rangle|^2$, the yellow line corresponds to $F_{triv}^{UB}(\lambda)$ and the green line corresponds to $F_{1|2345}^{max}(\lambda).$  As $\lambda$ approaches 1/2, if one chooses $\epsilon=1-F_{prot}(\lambda)$, then $\epsilon$ becomes sufficiently small that all states in the initial vicinity are isolated. Moreover, we see that, in this limit, $F_{prot}(\lambda)$ (the red line) is greater than the fidelity achievable with deterministic transformations. Therefore, there is a $\lambda<1/2$ and $\epsilon=1-F_{prot}(\lambda)$ such that an approximate transformation is possible, but there are no deterministic transformations between the $\epsilon$-vicinities.}
\label{Fig:ProtocolFidelity}
\end{figure}

With regards to (a), we want to ensure the initial and final $\epsilon$-vicinities (up to LUs) do not overlap, i.e., there is no state $\epsilon$-close to both $\ket{\phi}$ and $\ket{\psi}$ (up to LUs). To show that this is indeed the case, we can use the following upper bound (see Appendix \ref{sec:apppendixoverlapbetweenballs}):
\begin{align}
    \max_{U_i, V_i}  & \max_{\ket{\chi}} \min \left\{ F(\otimes_i U_i \ket{\psi},\ket{\chi}),\ F(\ket{\chi},\otimes_i V_i \ket{\phi(\lambda)}) \right\}   \nonumber\\
    &\le \frac{3}{4} +  \frac{1}{8} \left(\sum_i \sqrt{\mu_i(\rho_1(\lambda))}\right)^2 \equiv F_{triv}^{UB}(\lambda),
    \label{eq:maxoverlapofLUballs}
\end{align}
where $\mu_i(\rho_1(\lambda))$ are the eigenvalues of the reduced density matrix on qubit 1 for $\ket{\phi(\lambda)}$. If $F_{prot}(\lambda)>F_{triv}^{UB}(\lambda)$, then for $\epsilon=1-F_{prot}(\lambda)$, there is no state that is in both $\epsilon$-vicinities. This upper bound converges to $7/8\ll 1-\epsilon$ as $\lambda$ approaches $\frac{1}{2}$ (see Fig. \ref{Fig:ProtocolFidelity}), thus proving that, for large $\lambda$, there are no LU transformations between the initial and final $\epsilon$-vicinities. 

Next we must show (b): that there are no transformations from states $\epsilon$-close to the initial state to non-fully-entangled states $\epsilon$-close to the final state. The task of categorising LOCC transformations from fully-entangled states to non-fully-entangled states is challenging (see the preliminaries), let alone all such transformations from states in the initial $\epsilon$-vicinity. So instead, we consider a necessary condition for such a transformation. Namely, we look for the closest non-fully-entangled state to $\ket{\phi(\lambda)}$. If there are no non-fully-entangled states with fidelity greater than $F_{prot}(\lambda)$, then there are certainly no deterministic transformations from the initial $\epsilon$-vicinity to non-fully-entangled states with fidelity greater than $F_{prot}(\lambda)$ \footnote{As an alternate necessary condition, one can consider whether one can reach an ensemble of SLOCC equivalent four qubit states from the initial state. However, for this initial state, one can indeed show that such transformations are possible. Therefore, this necessary condition cannot be used to rule out deterministic LOCC transformations to four qubit states.}.

As the target state is permutationally symmetric, it is sufficient to consider two cases of non-fully-entangled states: (i) states which are not entangled in the partition $1|2345$ and (ii) states which are not entangled in the partition $12|345$, i.e., it is sufficient to consider \cite{VidalJonathanNielsen2000_ApproxLOCC}
\begin{align}
    F^{max}_{1|2345}(\lambda) &\equiv \max_{\ket{e_1}_{1}\ket{e_2}_{2345}} F(\ket{e_1}_{1}\ket{e_2}_{2345}, \ket{\phi(\lambda)})\\
                              & = \mu_{max}(\rho_1(\lambda)),\\
    F^{max}_{12|345}(\lambda) &\equiv \max_{\ket{e_1}_{12}\ket{e_2}_{345}} F(\ket{e_1}_{12}\ket{e_2}_{345}, \ket{\phi(\lambda)})\\
                              & =  \mu_{max}(\rho_{12}(\lambda)),
\end{align}
where $\mu_{max}(\rho_1(\lambda))$ ($\mu_{max}(\rho_{12}(\lambda))$) is the maximum eigenvalue of the one-(two)-qubit reduced density matrix of $\ket{\phi(\lambda)}$. All other partitions are combinations of refinements and permutations of these partitions. It is easy to verify that $F^{max}_{1|2345}>F^{max}_{12|345}$.

In summary, provided $F_{prot}(\lambda)$ is greater than $F_{triv}^{UB}(\lambda)$ and 
$F_{1|2345}^{max}(\lambda)$, then we may choose $\epsilon=1-F_{prot}(\lambda)$ and ensure that (a) there are no LU transformations between the $\epsilon$-vicinities, and (b) there are no transformations to non-fully-entangled states between the $\epsilon$-vicinities. In this case, provided $\epsilon$ is small enough ($F_{prot}(\lambda)$ is large enough) to ensure the initial $\epsilon$-vicinity consists of only generic states, we can rule out the possibility of a deterministic transformation between the $\epsilon$-vicinities.

We now complete this argument demonstrating that approximate transformations really are more than simply deterministic transformations from the $\epsilon$-vicinities around our states. Fig. \ref{Fig:ProtocolFidelity} shows the protocol fidelity, $F_{prot}(\lambda)$, plotted with $ F_{triv}^{UB}(\lambda)$ and $F_{1|2345}^{max}(\lambda) $. As mentioned, $F_{prot}(\lambda)>F_{triv}^{UB}(\lambda)$ for large $\lambda$. We have numerically verified that $F_{prot}(\lambda)>F^{max}_{1|2345}(\lambda)$ at a resolution of $10^{-2}$ across the range $(0,1/2)$, i.e., $\forall \lambda \in \{10^{-2},2\times 10^{-2},...,1/2-10^{-2}\}$. Moreover, we have also verified this holds in the vicinity of $\lambda=1/2$ (which is the limit of interest) up to a resolution of $10^{-10}$. This provides strong numerical evidence that $F_{prot}(\lambda)>F^{max}_{1|2345}(\lambda)$ $\forall \lambda \in (0,1/2)$. Assuming this inequality does hold for all $\lambda$  in the interval $(0,1/2)$, then we can freely choose $\lambda<1/2$ large enough and $\epsilon=1-F_{av}(\Lambda_\lambda(\psi),\phi(\lambda))$ such that the initial $\epsilon$-vicinity consists only of generic states, the final $\epsilon$-vicinity consists of only fully-entangled states, there are no LU transformations between the epsilon vicinities and $F_{prot}(\lambda)\ge 1-\epsilon$. Thus, by the above argument, we have an example of a faithful ensemble transformation, for which there is no $\epsilon$-close deterministic pure-state transformation.

To sum up, we have demonstrated strong numerical evidence that there does indeed exist a $\ket{\psi},\ket{\phi}$ and $\epsilon>0$ such that
\begin{equation}
    T_{ens}^{\epsilon,\epsilon}(\psi,\phi) \supsetneq T_{det}^{\epsilon,\epsilon}(\psi,\phi) = \emptyset.
\end{equation}
This demonstrates that approximate transformations are more powerful than simply deterministic LOCC transformations between states in the $\epsilon$-vicinities around the initial and final state.

\section{Conclusion}

Approximate LOCC transformations arise naturally as any real implementation of a quantum process in a lab is subject to noise. Thus, states are never pure, and any map can only be implemented within experimental error. Furthermore, in the multipartite case this study is also very relevant from the theoretical point of view given that in general almost every pure state is isolated. This means that it is generically impossible to transform a pure state into another with perfect fidelity. Thus, given a pair of input and target pure states, this raises the question of whether a faithful transformation with a reasonably large fidelity is possible. 

This problem has already been studied in the bipartite case, where the optimal fidelity has been determined for any pair of pure states. There it has been shown that the optimal fidelity is achieved by a deterministic protocol. However, an analogous result in the multipartite case seems completely out of reach. In the multipartite case, even transformations that can be achieved with fidelity equal to 1 are yet to be fully characterized outside of the generic case. Things only get more complicated in the approximate case due to the fact that the pure states in the target ensemble might belong to different SLOCC classes (including non-fully-entangled ones), which have an extremely complex structure in the multipartite domain. However, we have argued that certain simplifications take place when considering faithful LOCC transformations in this setting. In particular, we have shown that it is enough to consider pure input and output states when considering transformations from and to noisy mixed states in the vicinity of pure states, that finite round LOCC protocols approximate faithful transformation arbitrarily well, and that in fact SLOCC class preserving, finite-round LOCC protocols approximate faithful transformations arbitrarily well. In addition to this, when considering approximate transformations, no state is isolated and the MES is finite. 

Despite the aforementioned difficulty, in this paper we have aimed to study the general landscape of approximate transformations in the multipartite setting. For this, we have introduced a hierarchy of approximate transformations that are well motivated by both operational reasons and from the perspective of entanglement theory; namely, ensemble, optimal-conclusive, deterministic, ensemble-SLOCC and deterministic-SLOCC transformations. Our main result is that while this hierarchy is trivial in the bipartite case (i.e.\ $F_{ens}^\delta(\psi\rightarrow \phi) = F_{det}^\delta(\psi\rightarrow \phi)=F_{ens-SLOCC}^\delta(\psi\rightarrow \phi) = F_{det-SLOCC}^\delta(\psi\rightarrow \phi)$ for any fully-entangled, pure-state pair $\ket{\psi}$ and $\ket{\phi}$), this is not so in the multipartite case, thus revealing a fundamentally richer structure. In more detail, in Theorem 4 we have proven that $F_{ens}^\delta(\psi\rightarrow \phi) =F_{ens-SLOCC}^\delta(\psi\rightarrow \phi)$ for any pair of pure states $\ket{\psi}$ and $\ket{\phi}$ and any $\delta\geq0$, and we have used this and the properties of generic states to establish that the general inequalities $F_{ens-SLOCC}^0(\psi\rightarrow \phi) \geq F_{det-SLOCC}^0(\psi\rightarrow \phi)$ and $F_{ens}^0(\psi\rightarrow \phi) \geq F_{det}^0(\psi\rightarrow \phi)$ can be strict (with both results holding as well for sufficiently small values of $\delta>0$). Finally, we have considered the question of whether there exists an $\epsilon>0$ and a pair of pure states $\ket{\psi}$ and $\ket{\phi}$ such that $F^\epsilon_{ens}(\psi\rightarrow \phi)\geq1-\epsilon$ but $F^\epsilon_{det}(\psi\rightarrow \phi)<1-\epsilon$. The motivation for this question is two-fold. First, it models the practical case where all states can be determined up to a fixed resolution $\epsilon$. Second, if the above is true, this implies that optimal transformations do not need to arise by adding noise to a deterministic transformation. In this respect, we have considered a particular pair of states, derived an upper bound on $F^\epsilon_{det}(\psi\rightarrow \phi)$ and constructed an ensemble-transformation type of protocol that we have numerically verified violates the previous bound for any choice of $\epsilon$  that is sufficiently close to 0. 

In addition to the above, in the future it would be interesting to study further what the optimal fidelities can be for the different types of approximate transformations. While obtaining the optimal fidelities in general appears to be a formidable problem, it is worth studying whether reasonably sharp bounds can be found efficiently. Furthermore, it could be interesting to study faithful transformations of particularly physically-relevant multipartite states, for instance stabilizer states, matrix product states or non-generic states with known physical applications.

\section{Acknowledgements}
DG, MH, CS, and BK  acknowledge financial support from the SFB BeyondC (Grant No. F7107-N38), the Austrian Science Fund (FWF) through the grants P 32273-N27 and DK-ALM W1259-N27. DG and BK acknowledge the BMW endowment fund. JIdV acknowledges financial support from the Spanish Ministerio de Ciencia e Innovaci\'on (grant PID2020-113523GB-I00 and "Severo Ochoa Programme for Centres of Excellence" grant CEX2019-000904-S funded by MCIN/AEI/10.13039/501100011033) and from Comunidad de Madrid (grant QUITEMAD-CM P2018/TCS-4342 and the Multiannual Agreement with UC3M in the line of Excellence of University Professors EPUC3M23 in the context of the V PRICIT)

\bibliographystyle{apsrev4-1}

\section{Appendix}

\subsection{SEP Ensemble Transformations}
\label{sec:AppendixSepEnsTransfo}
Here we derive Thm. \ref{thm:SEPensemble}. The proof is analogous to that presented in Refs. \cite{Gour2011_SEP, GourKrausWallach2017_AlmostAllTrivStab, HebenstreitEtAl2021_SEP1isnotSEP}.

\setcounter{theorem}{1}
\begin{theorem}[\cite{Gour2011_SEP,HebenstreitEtAl2021_SEP1isnotSEP}]
The state $g\ket{\psi_s}$ can be transformed to the (finite) ensemble $\{(p_i, h_i\ket{\psi_s})\}$ (with $h_i$ local and invertible) via SEP if and only if there exists a finite set of probabilities $\{p_{ij}\}$, symmetries $\{S_j\} \subseteq \mathcal{S}_{\ket{\psi_s}}$, and $N_q\in\mathcal{N}_{g\ket{\psi_s}}$ such that $\sum_j p_{ij} = p_{i}$ and
\begin{equation}
    \sum_{ij} \frac{1}{r_{i}} p_{ij} S_{ij}^\dagger H_i S_{ij} + g^\dagger\sum_q N_q^\dagger N_q g = G,\label{eq:SEPensAppendix}
\end{equation}
where $r_i=||h_i\ket{\psi_s}||^2/||g\ket{\psi_s}||^2$.
\end{theorem}

\begin{proof}
The state $g \ket{\psi}/n_g$ can be transformed to the ensemble $\{(p_i, h_i \ket{\psi}/n_{h_i})\}$ iff there is a (CPTP) SEP channel $\Lambda$, which decomposes into trace-non-increasing separable CP maps, $\Lambda_i$, corresponding to each output; i.e., $\Lambda=\sum_i \Lambda_i\otimes \ket{i}\bra{i}$, such that 
\begin{equation}
    \Lambda_i\left( \frac{g| \psi\rangle\langle \psi |g^\dagger}{ n_g^2} \right)= p_k  \frac{h_i |  \psi\rangle\langle \psi | h_i^\dagger}{ n_{h_i}^2} \ \forall i.
    \label{eq:appendixSEPens1}
\end{equation}
For the same reasons as in Refs \cite{Gour2011_SEP, GourKrausWallach2017_AlmostAllTrivStab, HebenstreitEtAl2021_SEP1isnotSEP}, we have

\begin{equation}
    \Lambda_i(\cdot)=\sum_{j} M_{ij} (\cdot) M_{ij}^\dagger + \sum_q N_{iq}^\dagger (\cdot) N_{iq},
    \label{eq:appendixsepchannel}
\end{equation}
where $M_{ij}$ are invertible, $N_{qi}\in \mathcal{N}_{g \ket{\psi}_s}$ and
\begin{align}
    M_{ij} g \ket{\psi}/n_g &= \sqrt{p_{ij}} h_i \ket{\psi}/n_{h_i} \\
    N_{iq} g \ket{\psi}/n_g &= 0,
\end{align}
where $\sqrt{p_{ij}}$ are some positive coefficients. Rearranging the first equation above we have:
\begin{align}
    \frac{1}{\sqrt{p_{ij}}} \frac{n_{h_i}}{n_g} h_{i}^{-1}M_{ij}g \ket{\psi} = \ket{\psi}.
\end{align}
Hence, the operator belongs to the stabilizer of $\ket{\psi}$, and thus we have:
\begin{align}
    M_{ij}=\frac{n_g}{n_{h_i}} \sqrt{p_{ij}}\ h_i S_{ij} g^{-1}.
\end{align}
Applying this to Eq. (\ref{eq:appendixSEPens1}), we have $\sum_{j} p_{ij} = p_j$. Moreover, applying the completeness relation (and gathering annihilating operators) we end up with Eq. (\ref{eq:SEPensAppendix}).

\end{proof}

 As final comments, we note that it follows from the same arguments found in Ref. \cite{HebenstreitEtAl2021_SEP1isnotSEP} that $\text{LOCC}_{\mathbbm{N}}$ ensemble transformations within an SLOCC class are a subset of SEP$_1$ ensemble transformations. That is, any LOCC$_\mathbf{N}$ ensemble transformation within an SLOCC class does not make use of non-invertible, local operators (the $N_{iq}$ in Eq (\ref{eq:appendixsepchannel})). Moreover, it is easy to verify that the entanglement monotones introduced in Ref. \cite{Sauerwein2018_DifferentiableTransfo} (Eq. (\ref{eq:genericentmonotones})) are in fact invariant (monotonic) under deterministic (ensemble) SEP transformations within a generic SLOCC class. Whereas it was shown in 
 Ref. \cite{Sauerwein2018_DifferentiableTransfo} that they are invariant under SEP$_1$ within an SLOCC class, it is straightforward to show that (taking the non-invertible operators into account) any SEP map (with outputs only within the same SLOCC class) can never increase those entanglement monotones.

\subsection{Alternative Proof of Lemma 3}
\label{sec:AppendixusingVidalsInequality}
Here, we present an alternate proof that it is in fact sufficient to restrict ourselves to transformations where the initial state is pure. To see this, assume an initial mixed state, $\rho$, is $\delta$-close to some pure state, $\ket{\psi}$, and $\Lambda(\rho)$ is $\epsilon$-close to $\ket{\phi}$, with $\Lambda$ being some LOCC protocol. Then by Eq. (\ref{eq:vidalinequality}) we have that
\begin{equation}
    d(\psi \rightarrow \phi ) \le d(\rho \rightarrow \phi ) + d(\psi,\rho).
\end{equation}
Using Eq. (\ref{eq:FidInequality}) and Eq.  (\ref{eq:FidInequalityonepurestate}), we have 
\begin{equation}
    F(\psi \rightarrow \phi ) \ge 1 - ( \sqrt{\delta} + \sqrt{\epsilon}).
\end{equation}
Note, however that the bound in Lemma 3 is tighter as  $\epsilon \ll \sqrt{\epsilon}$.

\subsection{Approximate Transformations in the Limit $\epsilon\rightarrow 0$}
\label{sec:Appendixlimitingcase}

In this section, we consider multipartite approximate transformations in the limiting case where $\delta =0$ and $\epsilon$ is arbitrarily small. That is, we fix the initial state (set $\delta=0$) and ask which states can be reached via approximate transformations for all $\epsilon >0$, i.e., we consider
\begin{align}
    A(\ket{\psi}) = \{ \ket{\phi} \in \mathcal{H} :  \forall \epsilon > 0,\ \exists\ \Lambda\in T^{0, \epsilon}_{ens}(\psi,\phi) \}.
\end{align}
It is easy to see that $A(\ket{\psi})$ can be equivalently characterised as the set of  limit points of a sequence of LOCC maps evaluated on $\ket{\psi}$, i.e., 
\begin{align}
    A(\ket{\psi}) &= \big\{\ket{\phi}\in\mathcal{H} : \exists (\Lambda_i)_{i\in\mathbbm{N}}\subseteq LOCC \nonumber \\
    & \qquad \qquad \quad  : \lim_{i\rightarrow \infty}\ \Lambda_i(|\psi\rangle\langle \psi|) \ =\ |\phi\rangle\langle \phi|  \big\},
\end{align}
where $\lim_{i\rightarrow \infty}\ \Lambda_i(|\psi\rangle\langle \psi|) \ =\ |\phi\rangle\langle \phi|$ denotes that  $\Lambda_i(|\psi \rangle\langle\psi |)$ converges to $|\phi\rangle\langle\phi|$ with respect to the trace distance. 

All $\text{LOCC}$ maps are  separable maps. As the set of separable maps, $\text{SEP}$, on a finite Hilbert space is a closed and bounded subset of a finite vector space, it is compact. Therefore, by the Bolzano–Weierstrass theorem \cite{Bartle1999_BolzanoWeierstrass}, the sequence $(\Lambda_i)_{i\in\mathbbm{N}}\subseteq \text{LOCC} \subseteq \text{SEP}$ will contain a sub-sequence, $(\tilde{\Lambda}_i)\subseteq \text{SEP}$, which converges to some separable map, $\mathcal{E}\in \text{SEP}$, i.e., $\tilde{\Lambda}_i \rightarrow_\diamond \mathcal{E}\in \text{SEP}$ (where convergence is defined with respect to the diamond norm). Convergence wrt to the diamond norm implies that, for all $\ket{\psi}$, $\lim_{i\rightarrow \infty}\ \tilde{\Lambda}_i(|\psi\rangle\langle \psi|)\ =\ \mathcal{E}(|\psi\rangle\langle \psi|)$. In particular, as $\lim_{i\rightarrow \infty}\ \tilde{\Lambda}_i(|\psi\rangle\langle \psi|) \ =\ |\phi\rangle\langle \phi|$, we have $ \mathcal{E}(|\psi\rangle\langle \psi|) =|\phi\rangle\langle \phi|$. That is, every element of $A(\ket{\psi})$ can also be reached from $\ket{\psi}$ via a separable map, i.e.,
\begin{align}
    A(\ket{\psi})\subseteq \big\{\ket{\psi}\in\mathcal{H} : \exists\ \mathcal{E}\in \text{SEP}, \ \mathcal{E}(|\psi\rangle\langle \psi|) =|\phi\rangle\langle \phi|\big\}.
\end{align}
This result tells us that the limiting case of approximate transformations is at most as powerful as $\text{SEP}$ transformations. Consequently, all the result regarding isolation and the MES from Ref. \cite{DeVicenteEtAl2013_MES, GourKrausWallach2017_AlmostAllTrivStab, SauerweinEtAl2018_AlmostAllStatesNotReachable} (see the preliminaries) also hold in this limiting case.

\subsection{Overlap between $\epsilon$-Vicinities}
\label{sec:apppendixoverlapbetweenballs}

In this appendix, we study when the vicinities of the initial and the final state can overlap. In this case, the trivial transformation, i.e., not doing anything, would transform the initial state to the target state. Note that the fact that we have to consider any LU-equivalent state here complicates the derivation of conditions which ensure that no trivial transformation exists. 

Let us define

    \begin{align}
        &F_{triv}(\ket{\psi},\ket{\phi}) \nonumber\\
        &\quad \equiv \max_{U_i, V_i}  \max_{\ket{\chi}} \min \left\{ F(\otimes_i U_i \ket{\psi},\ket{\chi}),\ F(\ket{\chi},\otimes_i V_i \ket{\phi(\lambda)}) \right\} \nonumber \\
        &\quad = \max_{U_i}  \max_{\ket{\chi}} \min \left\{ F(\ket{\psi},\ket{\chi}),\ F(\ket{\chi}, \otimes_i U_i  \ket{\phi(\lambda)}) \right\}.
    \end{align}
    
Given an $\epsilon\ge0$, if $F_{triv}\ge 1-\epsilon$, then there exists a state, $\ket{\chi}$, that lies in both the $\epsilon$-vicinities of $\ket{\psi}$ and $\ket{\phi}$, and therefore there is a trivial transformation between the two states. Conversely, if we construct a protocol such that $F_{prot}>F_{triv}$, then we know that, for $\epsilon=1-F_{prot}$, there is no state which lies in both $\epsilon$-vicinities at the same time. Hence, the trivial transformation, i.e., doing nothing, would not lead to a higher fidelity. We then have the following observation.

\begin{observation}
   The following bounds on $F_{triv}(\ket{\psi}\ket{\phi})$ hold
    \begin{align}
        F^{LB}_{triv}(\ket{\psi},\ket{\phi})\le F_{triv}(\ket{\psi},\ket{\phi}) \le F^{UB}_{triv}(\ket{\psi},\ket{\phi}),
    \end{align}
    where
    \begin{align}
    F_{triv}^{LB}(\ket{\psi},\ket{\phi})&=F(\ket{\psi},\ket{\phi}),\\
    F_{triv}^{UB}(\ket{\psi},\ket{\phi})&= \frac{3}{4} +  \frac{1}{4} \left(\sum_i \sqrt{\mu_i(\rho_1(\psi))\mu_i(\rho_1(\phi))}\right)^2 
    \end{align}
    where $\mu_i(\rho_1(\psi))$ ($\mu_i(\rho_1(\phi))$) are the ordered eigenvalues of the site-1 reduced density of $\ket{\psi}$ ($\ket{\phi}$). 
\end{observation}

Before proving this theorem, note that one need not only consider the site-1 reduced density matrix to obtain a bound. Any bipartite splitting leads to an upper bound. We chose to optimise over bipartite unitaries in the bipartition $1|2...n$ because this suffices for our examples.

\begin{proof}
    Whereas the lower bound is trivial, the upper bound can be shown as follows. We begin by noting that, as $\ket{\psi}, \ket{\phi}$ and $\ket{\chi}$ are all pure, the upper bound in Eq. (\ref{eq:FidInequality}) is an equality. Therefore, we have
\begin{align}
    \min &\{F(\ket{\psi},\ket{\chi}), F(\ket{\chi},\ket{\phi})    \nonumber \\
    &= \min \left\{ 1-D^2(\ket{\psi},\ket{\chi}), 1-D^2(\ket{\chi},\ket{\phi})\right\}.
\end{align}
Now, although $F$ is not a metric, $D$ is. Therefore, we have
\begin{align}
    D(\psi,\phi)\le D(\psi,\chi)+D(\chi,\phi) \le 2 \max \left\{D(\psi,\chi), D(\chi,\phi)    \right\}.
\end{align}
As $D$ is positive, we have
\begin{equation}
    \frac{1}{4}D^2(\phi,\psi) \le \max \left\{D^2(\psi,\chi), D^2(\chi,\phi)    \right\},
\end{equation}
and thus, we have for any $\ket{\chi}$
\begin{align}
    \min \{F(\ket{\psi},\ket{\chi}), F(\ket{\chi},\ket{\phi})   
    &\le 1- \frac{1}{4} D^2(\phi,\psi)\\
    &= \frac{3}{4} + \frac{1}{4} F(\ket{\phi},\ket{\psi}).
\end{align}
Therefore, we have
\begin{align}
&F_{triv}(\ket{\psi},\ket{\phi})\\
    &\quad \le \frac{3}{4} +\frac{1}{4} \max_{U_i} F(\ket{\psi}, \otimes_i U_i \ket{\phi}) \nonumber\\
    &\quad \le \frac{3}{4} +\frac{1}{4} \max_{U,V} F(\ket{\psi}, U\otimes V \ket{\phi}) \nonumber \\
    &\quad = \frac{3}{4} +  \frac{1}{4} \left(\sum_i \sqrt{\mu_i(\rho_1(\psi))\mu_i(\rho_1(\phi))}\right)^2 \nonumber\\
    &\quad \equiv F_{triv}^{UB}(\ket{\psi},\ket{\phi}),
     \label{eq:Fnothingupperbound}
\end{align}
where, in the third to last line, we have used bipartite unitaries instead of fully-local unitaries to provide an upper bound, and the second to last line follows from Ref. \cite{VidalJonathanNielsen2000_ApproxLOCC}, with $\mu_i(\rho_1(\psi))$ corresponding to the sorted eigenvalues of the qubit-1 reduced density matrix of $\psi$ (and likewise for $\phi$).
\end{proof}

\end{document}